 \documentclass[journal]{IEEEtran}
\usepackage{amsmath,amssymb,amsfonts}
\usepackage{mathtools}
\usepackage{amsthm}
\usepackage{mathdots}
\usepackage{graphicx}
\usepackage{xcolor}
\usepackage{textcomp}
\usepackage{setspace}
\usepackage[style=ieee, doi=false]{biblatex}
\usepackage{mycommands}
\usepackage{algpseudocode,algorithm}
\usepackage{hyperref}
\usepackage{array}
\usepackage{booktabs,multirow,makecell}
\usepackage{scalerel}

\usepackage{bm}

\newtheorem{proposition}{Theorem}
\newtheorem{corollary}{Corollary}

\ProvideDocumentCommand{\intinf}{}{\int_{-\infty}^{\infty}\!\!}
\ProvideDocumentCommand{\intzeroinf}{}{\int_{0}^{\infty}\!\!}
\RenewDocumentCommand{\bard}{}{\,\,\bar{\!\!d}}

\makeatletter
\AtBeginDocument{\DeclareMathVersion{bold}
	\SetSymbolFont{operators}{bold}{T1}{times}{b}{n}
	\SetMathAlphabet{\mathrm}{bold}{T1}{times}{b}{n}
	\SetMathAlphabet{\mathit}{bold}{T1}{times}{b}{it}
	\SetMathAlphabet{\mathbf}{bold}{T1}{times}{b}{n}
	\SetMathAlphabet{\mathtt}{bold}{OT1}{pcr}{b}{n}
	\SetSymbolFont{symbols}{bold}{OMS}{cmsy}{b}{n}
	\renewcommand\boldmath{\@nomath\boldmath\mathversion{bold}}}
\makeatother

\AtBeginBibliography{\footnotesize}
\hypersetup{
	colorlinks=true,
	linkcolor=blue,    
	urlcolor=blue,
	citecolor=blue,
	citebordercolor=blue,
	linkbordercolor=blue,
	pdftitle={C\&R-RBP},
}

\begin{document}

\title{Low-Latency LDPC Decoding Based on Consensus and Return Residual Belief Propagation}
\author{Allan~E.~Feitosa and Rodrigo~C.~de~Lamare \vspace{-2em}
\thanks{Departament of Electrical Engineering, PUC-Rio, Rio de Janeiro 22451-900, Brazil and 	University of York, York, UK.}}

\maketitle

	
\begin{abstract}
We propose two new decoding algorithms for low-density parity-check (LDPC) codes. The \emph{Consensus} and \emph{Return} residual belief propagation (C\&R-RBP) algorithm reduces the number of iterations required to achieve a target decoding error rate over an AWGN channel by exploiting the availability of updated information, thereby mitigating the performance degradation caused by the greedy behavior of RBP-based algorithms and reducing error propagation. The \emph{Consensus} and \emph{Delayed Return} RBP algorithm (C\&DR-RBP) achieves better performance in the early decoding iterations than C\&R-RBP, making it particularly suitable for low-latency applications.
Both proposed methods maintain the same computational complexity as the original RBP algorithm in terms of residual calculations.
We also provide a formal analysis showing that the \emph{Consensus} mechanism increases the probability of correct bit decisions under initial decoding conditions and in the presence of severe bit transmission errors, providing insight into the performance gains observed in the simulations.
Numerical results based on the rapidly growing 5G New Radio (NR) standard, considering different code lengths, code rates, and channel models, demonstrate that both methods outperform competing RBP-based algorithms with similar overall computational complexity in terms of frame error rate (FER) and bit error rate (BER). 
\end{abstract}
\begin{IEEEkeywords}
LDPC codes, Residual Believe Propagation, RBP, message passing, dynamic scheduling techniques, low decoding latency, 5G NR\@.
\end{IEEEkeywords}

\section{Introduction}\label{sec:introduction}
        Low-density parity-check (LDPC) codes have gained widespread attention due to their Shannon capacity approaching performance. This is evidenced by their adoption as the channel coding scheme in several communication standards, including WiFi \cite{WiFi2004} and WiMAX \cite{WiMAX2005}.
More recently, LDPC codes were selected by the 3rd Generation Partnership Project as the data channel coding for the 5G new radio (NR) enhanced Mobile Broadband (eMBB) \cite{3GPP2025,Richardson2018}.

LDPC codes were first introduced by Gallager \cite{Gallager1962} in 1962, together with a near-optimal decoding algorithm known as the \emph{sum–product algorithm} (SPA), also referred to as \emph{belief propagation} (BP) in the literature.
However, Gallager's work received limited attention for several decades, with the notable exception of Tanner's work in 1981 \cite{Tanner1981}, which generalized LDPC codes and introduced their graphical representation, the \emph{Tanner graph}.
A Tanner graph is a bipartite graph composed of variable nodes (VNs), representing coded input bits, and check nodes (CNs), representing parity-check equations.
In 1999, MacKay \cite{MacKay1999} contributed significantly to the revival of LDPC codes and described the classical implementation of the SPA algorithm, the so-called \emph{flooding} scheduling, in which the BP algorithm performs \emph{message-passing} between CNs and VNs in the Tanner graph.
Information transmitted from CNs to VNs is referred to as a C2V message, whereas information transmitted in the opposite direction is referred to as a V2C message.
In flooding scheduling, all V2C messages are updated and propagated simultaneously, followed by the simultaneous update of all C2V messages.

Following MacKay's revival of LDPC codes, many strategies were proposed to improve code error-rate performance.
Whereas many strategies focus on optimizing the design of the Tanner graph, such as the progressive edge growth (PEG) designs \cite{peg,bfpeg,dopeg,memd}, this paper focuses on message scheduling algorithms that differ from the conventional flooding scheme.
One prominent example is \emph{sequential scheduling}, such as the Layered Belief Propagation (LBP) algorithm \cite{Hocevar2004}, which adopts a different strategy from flooding scheduling: for each CN \(c\muspace\), all V2C messages associated with \(c\) are updated first, followed by all the corresponding C2V messages.
A subsequent development was the class of the so-called \emph{informed dynamic scheduling} (IDS) methods \cite{Hocevar2004,Casado2010,kaids} and reweighting techniques \cite{vfap}, among which the \emph{Residual Belief Propagation} (RBP) \cite{Casado2010} is perhaps the most representative.
In RBP, the sequence of message updates is dynamically determined based on the C2V residual, defined as the absolute difference between C2V values before and after an update.
The C2V with the largest residual, represented by the pair (\(c^*\), \(v^*\)), is selected for update, followed by V2C messages from \(v^*\) to all CNs other than \(c^*\).
RBP achieves better performance than flooding and LBP scheduling schemes, albeit at the cost of increased computational complexity.

It is well known that RBP suffers from performance degradation due to its \emph{greedy behavior}.
As discussed in \cite{Casado2010}, the largest-residual criterion tends to prioritize message propagation towards the least reliable VNs.
Consequently, decoding errors may be propagated when a C2V update produces an incorrect bit estimate.
Furthermore, as reported in \cite{Lee2013}, when a received codeword cannot be successfully corrected, a small group of edges in the Tanner graph tends to be updated repeatedly, while other VNs have little or no opportunity to contribute their intrinsic information to the decoding process.
Although RBP often achieves lower error rates with a smaller number of message updates, its greedy nature may lead to degraded error-rate performance after many decoding iterations.
Similar observations regarding the greedy behavior of RBP have also been reported in \cite{Zhang2019,Chang2021,Yan2024}.

Subsequent RBP-based IDS algorithms sought to improve performance by mitigating the performance degradation of RBP caused by the greedy behavior.
The \emph{Node-wise} RBP (NW-RBP) \cite{Casado2010} updates and propagates all C2V messages corresponding to \(c^*\) to alleviate RBP greediness.
The \emph{Silent-Variable-Node-Free} RBP (SVNF) \cite{Lee2013} ensures that all VNs have the same opportunity to contribute their intrinsic information to the decoding process.
The Oscillating Variable Nodes RBP (OV-RBP) \cite{Liu2015} identifies unreliable VNs for prioritized update based on sign oscillations of the VN posterior log-likelihood ratio (LLR).
The Residual-Decaying RBP (RD-RBP) \cite{Zhang2019} penalizes repeatedly updated residuals through an attenuation factor.
The \emph{Conditional Innovation} RBP (CI-RBP) \cite{Chang2021} estimates the potential improvement of a VN using the difference between the current and precomputed conditional posterior probabilities of correct bit decision, albeit at the expense of increased computational complexity required for probability estimation.
In the \emph{Update Before Propagate} RBP (UBP-RBP) \cite{Chang2022}, after a CN transmits its C2V message, it remains muted until new information is received from its neighboring VNs.
Finally, the List-RBP \cite{Yan2024} employs two partially updated residual lists to prevent a small subset of residuals from monopolizing the update process.

Compared to traditional Flooding and LBP algorithms, RBP-based IDS methods can achieve a given error-rate performance in considerably fewer algorithm iterations.
However, this reduction comes at the cost of increased computational complexity per iteration.
Therefore, there is a trade-off that should be considered when evaluating actual decoding throughput and latency.
For example, conventional low-complexity approaches based on the \emph{min-sum} approximation of SPA \cite[p.~226]{Ryan2009}, such as layered min-sum or offset min-sum decoders, can offer lower computational complexity per iteration, although requiring more iterations to achieve the same performance of the original SPA implementation.

In this work, we propose two novel RBP-based algorithms, the \emph{Consensus and Return} RBP (C\&R-RBP) and the \emph{Consensus and Delayed Return} RBP (C\&DR-RBP) \cite{crrbp}, designed to achieve superior error-rate performance compared with existing RBP-based algorithms while maintaining the same computational complexity as the original RBP in terms of residual calculations.
This is achieved by two new techniques, namely the \emph{Consensus} and \emph{Return} mechanisms, which differ from the aforementioned RBP-based strategies and are designed to mitigate the performance degradation caused by the greedy behavior of RBP and reduce error propagation.
The \emph{Consensus} mechanism exploits the availability of updated information more effectively than existing RBP-based methods, while \emph{Return} is a complement of this strategy which allows further propagation of updated information across the code graph.
In particular, C\&DR-RBP applies the \emph{Return} mechanism only after a small number of algorithm iterations and requires fewer iterations than other RBP-based methods to achieve a given error-rate performance, as demonstrated by simulation results.
A detailed evaluation of the computational complexity of both algorithms is provided.
We also show an analysis showing that the consensus mechanism increases the probability of correct bit decisions under both initial decoding conditions and in the presence of severe bit transmission errors, providing insight into the performance gains observed by simulations.
Numerical results show that both algorithms outperform competing RBP-based methods with similar overall computational complexity, ensuring a fair comparison among different decoding strategies.
Simulations are conducted using LDPC parity-check matrices based on the 5G NR standard, reflecting its widespread adoption in modern communication systems.	

This paper is organized as follows.
Section~\ref{system} presents the system model, the channel characteristics, and problem statement.
Section~\ref{proposed} introduces the proposed methods, an overview of message-passing algorithms for LDPC decoding and residual belief propagation.
Section~\ref{Analysis} provides the computational complexity analysis of the proposed methods and a formal analysis of \emph{Consensus} under initial decoding conditions.
Section~\ref{considerations} further discusses some aspects of 5G NR codes.
Section~\ref{Implementation} discusses practical implementations and parallelization.
Section~\ref{Simul} presents our simulation results for different scenarios.
In Section~\ref{Latency} we estimate and compare the decoding latency and memory requirements of our proposed algorithms and other RBP-based methods.
Section~\ref{Conclusions} concludes this paper and discusses directions for future research.
	
\section{System model and problem statement}\label{system}
	
        In this section, we detail the LDPC coding system model, the channel characteristics, and the problem we address herein.
	
\subsection{Channel Model and Modulation}   

A typical digital communication system is depicted in Fig.~\ref{fig:system}.
A message is encoded at the source, producing the message bit vector \(\bma \in \sR^K\).
The channel encoder adds structured redundancy bits, represented by vector \(\bmr\), to enable forward error correction, generating the codeword \(\bmb \in \sR^N\), such that
\[ \bmb = \begin{bmatrix}
    \,\bma\, \\
    \,\bmr\,
\end{bmatrix}.\]
After encoding, \(\bmb\) is modulated using binary phase-shift keying (BPSK), mapping each bit \(b_j\) to a symbol \(u_j\) according to
\begin{equation}\label{bitmapping}
    u_j = (-1)^{b_j}, \ j \in \{1,2,...,N\}.
\end{equation}
This mapping associates bit \(0\) with the symbol \(+1\) and bit \(1\) with the symbol \(-1\), resulting in antipodal signaling suitable for binary transmission.
The received signal vector is given by \(\bmy = \bmu + \bmw\), where \(\bmu = [\,u_1 \ ... \ u_N]^\top\) and \(\bmw = [\,w_1 \ ... \ w_N]^\top\) is an additive white Gaussian noise (AWGN) process that models thermal noise and other random disturbances in the communication channel.
The AWGN model is widely adopted in the LDPC code literature and is suitable for many practical scenarios, including deep-space communication and point-to-point microwave links.
In addition, it serves as a baseline for more complex models.
Despite its limitations, such as neglecting impulsive noise, fading effects and channel correlation, the AWGN model remains attractive due to its analytical tractability and the relative simplicity for parameter estimation compared with non-Gaussian models.
Therefore, the theoretical analysis in this work is developed under the AWGN assumption, which provides a well-established baseline for analyzing the proposed RBP-based IDS algorithms.
\begin{figure}[ht]
\vspace{-1em}
    \centering
    \includegraphics[width=1\linewidth]{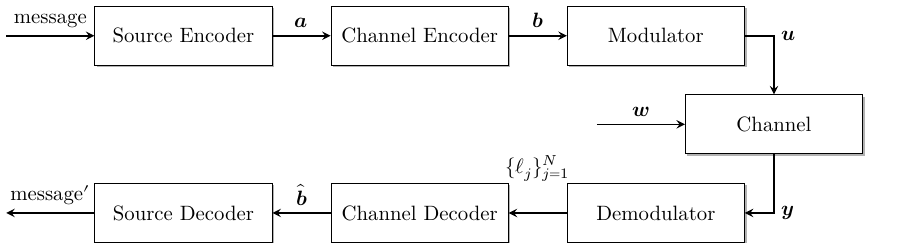}
    \caption{Block diagram of a typical digital communication system.}
    \label{fig:system}
\end{figure}   

At the receiver, the demodulator processes the received signal vector \(\bmy\) and computes soft log-likelihood ratios (LLRs) for each received symbol.
Considering the received signal
\begin{equation}\label{yj}
    y_j = u_j + w_j\,,
\end{equation}
BPSK modulation over an AWGN channel, and the bit mapping defined in \eqref{bitmapping}, the LLR \(\ell_j\) associated with \(y_j\) is given by
\begin{equation}\label{LLR}
    \ell_j = \ln\left(\tfrac{f(y_j \mid b_j = 0)}{f(y_j \mid b_j = 1)}\right) = \tfrac{2y_j}{\sigma^2},
\end{equation}
where \(\sigma^2 = \Var{w_j}\) \(\forall j\muspace\), and \(f(y_j \mid b_j = b)\), denotes~the~probability density function (PDF) of the random variable \(Y_j\muspace\), conditioned on \(b_j = b\), \(b \in \{0,1\}\), and evaluated at \(y_j\muspace\).
These LLRs provide the initial soft information for the iterative decoding process, which estimates the transmitted codeword \(\bmhatb\) by correcting errors introduced during transmission.
Finally, the source decoder reconstructs the original message.

\subsection{LDPC codes}

LDPC codes are \((N,K)\) linear block codes characterized by a sparse \(M\!\times\!N\) parity check matrix \(\bmH\), where \(M\!=\!N-K\).
Each row of \(\bmH\) represents a single parity-check equation.
The codeword \(\bmb\) must satisfy \(\bmH \bmb = \symbfit{0}\) in GF(2).
Any LDPC code can be represented by a bipartite graph \(\cG = (\cV,\cC,\cE)\), known as a Tanner graph, where \(\cV = \{v_j\}_{j=1}^N\) denotes the set of VNs, each \(v_j\) representing a codeword bit \(b_j\muspace\), and \(\cC = \{c_i\}_{i=1}^M\) denotes the set of CNs, each \(c_i\) representing a parity-check equation.
Finally, \(\cE\) denotes the set of edges connecting VNs and CNs, with cardinality \(E~=~\abs{\cE}\).

LDPC codes can be classified as regular or irregular according to their \emph{node degree distributions}.
In regular LDPC codes, the node degrees are uniform, i.e., all variable nodes are connected to the same number of check nodes, and vice versa.
In irregular LDPC codes, however, these degrees are non-uniformly distributed among nodes.
The 5G NR LDPC codes considered in this work are inherently irregular.

\subsection{Problem Statement}

The fundamental challenge in LDPC decoding is to efficiently recover the transmitted codeword \(\bmb\) from noisy observations \(\bmy\) while minimizing (\textit{i}) the number of algorithm iterations required to achieve a target error rate and (\textit{ii}) the associated computational complexity.
The combined impact of iterative processing and computational complexity contributes significantly to decoding latency, which is a critical issue in real time applications.
Our objective is to develop an iterative LDPC decoding algorithms with low decoding latency by means of a reduced number of iterations, and capable of achieving reliable error correction, while maintaining similar overall computational complexity of existing RBP-based  algorithms.
	
	\section{Proposed C\&R-RBP and C\&DR-RBP Algorithms}\label{proposed}
	
        In this section, we review message-passing algorithms and dynamic scheduling strategies and describe the proposed C\&R-RBP and C\&DR-RBP algorithms.
	
\subsection{Message-passing}

Both proposed algorithms are based on the exchange of information between VNs and CNs.
Information transmitted from CNs to VNs is called a C2V message, whereas information transmitted in the opposite direction is called a V2C message.
Consequently, the number of distinct messages is equal to the number of edges \(E\) in the Tanner graph.
Following the log-domain formulation proposed in \cite{Xiao2001}, the V2C message \(\vc{j}{i}\) associated with an edge \((c_i,v_j)\) is computed as
\begin{equation}\label{v2c}
    \vc{j}{i} = \textstyle \sum_{c_a \in \,\cC(v_{j})\setminus \{c_i\}} \cv{a}{j} + \ell_{j\muspace},
\end{equation}
where \(\cv{a}{j}\) denotes the C2V message received at the VN \(v_j\) from the CN \(c_a\muspace\), \(\cC(v_j)\!\setminus\!\{c_i\}\) represents the set of CNs linked to \(v_j\) excluding \(c_i\muspace\), and \(\ell_j\) denotes the LLR defined in \eqref{LLR}.

The C2V message \(\cv{i}{j}\) associated with an edge \((c_i,v_j)\) is updated according to
\begin{equation}\label{c2v}
    \cv{i}{j} = 2 \, \tanh^{-1}\Big(\textstyle \prod_{v_k \in \, \cV(c_i)\setminus \{v_j\}} \tanh\left(\frac{\vc{k}{i}}{2}\right)\!\!\Big),
\end{equation}
where \(\vc{k}{i}\) is the V2C message received at the CN \(c_i\) from the VN \(v_k\muspace\), and \(\cV(c_i)\!\setminus\!\{v_j\}\) is the set of VNs linked to \(c_i\) excluding \(v_j\muspace\).
Defining the \emph{posterior} LLR \({\ell^\prime}_{\!\!j}\) as
\begin{equation}\label{post}
    \ell_j^\prime = \sum_{c_i \in \, \cC(v_j)} \cv{i}{j} + \ell_j\,,
\end{equation}
and based on the bit mapping defined in \eqref{bitmapping}, the bit estimate \(\hatb_j\) is decided according to the test
\begin{equation}\label{estimate}
    \hatb_j = \begin{cases}
        1, & \text{if } \ell_j^\prime < 0, \\
        0, & \text{otherwise}.
    \end{cases}
\end{equation}   

The particular order in which the messages are updated determines the scheduling strategy of the decoding algorithm (e.g., Flooding, LBP, and RBP).
The update process is repeated until a stopping criterion is satisfied, namely, either all the parity check constraints are fulfilled or a maximum number of iterations is reached.

\subsection{Dynamic Scheduling using residuals}\label{RBP}

RBP-based methods dynamically determine the sequence of message updates based on the C2V \emph{residual}, defined as
\begin{equation}\label{res}
    \residual{i}{j} = |\,\cv{i}{j}^\prime - \cv{i}{j}|,
\end{equation}
where \(\cv{i}{j}^\prime\) denotes the value of \(\cv{i}{j}\) computed using the most up-to-date information available, i.e., the most recently updated V2C messages received at CN \(c_i\).
To illustrate the use of residuals in dynamic scheduling, consider the RBP procedure depicted in Figure~\ref{Fig:RBP}, which shows a subgraph of a Tanner graph in which both the VN and CN degrees are equal to 4.
The degree of a node corresponds to the number of neighboring nodes connect to it\footnote{For irregular Tanner graphs, i.e., graphs in which the number of connected nodes varies across the graph, the VN and CN degrees are defined as the average number of neighboring connected nodes.}.
The basic steps of RBP are summarized as follows:
\begin{enumerate}
    \item  Find the pair \((c_{i}^*,v_{j}^*)\) associated with the C2V message that has the largest residual among all the C2V messages, and update \(\cvopt\) according to \eqref{c2v} (depicted as the orange arrow in Figure~\ref{Fig:RBP});
    \item Update, using \eqref{v2c}, the V2C messages associated with all CNs connected to \(v_{j}^*\muspace\), excluding \(c_{i}^*\muspace\) (depicted as blue arrows in Figure~\ref{Fig:RBP});
    \item For each \(c_i\muspace\) obtained in the previous step, compute the new \(\cv{i}{j} ^\prime\) and residuals \(\residual{i}{j}\) using \eqref{res} for all VNs \(v_j\) connected to \(c_i\muspace\), excluding \(v_{j}^*\muspace\) (depicted as purple dotted arrows in Figure~\ref{Fig:RBP}).
\end{enumerate}

\begin{figure}[ht]
    \centering
    \includegraphics[width=\linewidth]{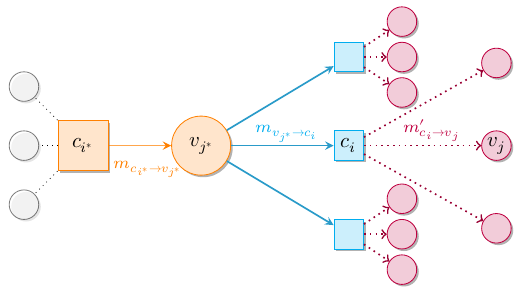}
    \caption{Illustration of the basic steps of the RBP algorithm on a Tanner graph subgraph with VN and CN degrees equal to 4.}
    \label{Fig:RBP}
\end{figure}

\begin{figure}[ht]
    \centering
    \includegraphics[trim={15 20 20 20}, clip,width=\linewidth]{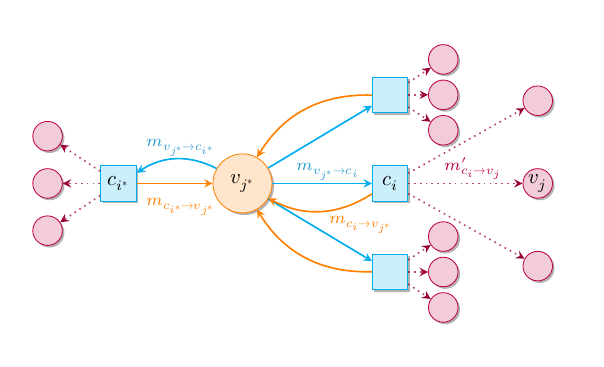}
    \caption{C\&R-RBP algorithm. \emph{Consensus} is represented by orange arrows. Blue arrows represent the default V2C updates and \emph{Return}. Purple dotted arrows represent the default RBP residual calculations plus the additional residual calculations due to \emph{Return}.}
    \label{Fig:CR-RBP}
\end{figure}	

These steps are repeated until the number of updated C2V messages reaches \(E\), the number of edges in the Tanner graph, which corresponds to the total number of distinct C2V messages.
Therefore, one decoding iteration consists of \(E\) C2V message updates.
The algorithm continues until a stopping criterion is satisfied.
Finally, the bit estimates are obtained according to \eqref{estimate}.

\subsection{Proposed C\&R-RBP Algorithm}

C\&R-RBP introduces two modifications to the conventional RBP procedure in order to improve the error-rate performance and reduce the number or iterations required to achieve a given error-rate performance.
To illustrate these strategies, consider the same Tanner subgraph example shown in Figure~\ref{Fig:RBP}, now reproduced in Figure~\ref{Fig:CR-RBP}.
The proposed strategies are summarized as follows:
\begin{enumerate}
    \item \emph{Consensus}: in addition to updating \(\cvopt\), the algorithm also updates the remaining C2V messages associated with the VN \(v_j^*\) corresponding to the largest residual (depicted as orange arrows in Figure~\ref{Fig:CR-RBP}).
    \item \emph{Return}: The returning V2C message \(\vcopt\) is also updated together with the other V2C messages from \(v_{j^*}\) (depicted as blue arrows in Figure~\ref{Fig:CR-RBP}).
\end{enumerate}

The motivation behind the \emph{Consensus} strategy is to reduce the bias introduced by the message \(\cvopt\muspace\) during the subsequent V2C update, which can be interpreted, from  \eqref{v2c}, as a weighted consensus operation.
If only \(\cvopt\) contributes new information, the resulting consensus becomes strongly biased towards this message.
By additionally updating the remaining C2V messages associated with VN \(v_j^*\muspace\), the proposed approach seeks to reduce the propagation of errors and mitigate the performance degradation of RBP\@.
In Subsection~\ref{Math}, we show through mathematical analysis that this intuition holds under the initial decoding conditions.	
Importantly, the additional C2V updates introduced by the \emph{Consensus} strategy do not increase computational complexity, since the messages \(\cv{i}{j} ^\prime\), computed during residual evaluation, can be stored and reused later.
The extra memory requirement is relatively small and is analyzed in Subsection~\ref{memory}.

The \emph{Return} strategy complements \emph{Consensus}.
Since C2V messages beyond \(\cvopt\) are updated, the returning V2C message \(\vcopt\) now carries additional information originating from the newly updated C2V messages.
Furthermore, the residuals corresponding to the C2V messages from \(c_i^*\) (represented by the leftmost purple dotted arrows in Figure~\ref{Fig:CR-RBP}) must also be computed.
To preserve the same computational complexity in terms of residual calculation as conventional RBP, the number of updates per iteration is reduced from \(E\) to \(E-N\).
This reduction is detailed in Section~\ref{Analysis}.

To further mitigate the greedy behavior, we incorporate the residual-decaying factor of \cite{Zhang2019} into the residual calculation:
\begin{equation}\label{decay}
    \residual{i}{j} = \beta_{i,j}\,\lvert \cv{i}{j}^\prime - \cv{i}{j} \rvert,
\end{equation}
where \(\beta_{i,j}\) is the residual attenuation factor.
At the start of each algorithm iteration, \(\beta_{i,j}\ = 1\), and each time \(\cv{i}{j}\) is selected for updating, \(\beta_{i,j}\) is multiplied by a constant \(\beta\), typically between \(0.8\) and \(0.9\).
At the end of each iteration, all \(\beta_{i,j}\) are reset to prevent residuals from being excessively attenuated.
Algorithm~\ref{alg1} summarizes the C\&R-RBP algorithm.

\begin{algorithm}[ht]
    \caption{C\&R-RBP decoding.}
    \begin{algorithmic}[1]
        \State Initialize \(e = 0\)
        \State Initialize all \(\beta_{i,\muspace j} = 1\), \(\cv{i}{j} = 0\), and \(\vc{j}{i} = \ell_j\)		
        \State Compute and store all \(\cv{i}{j}^\prime\)
        \State Set all \(\residual{i}{j} = |\cv{i}{j}^\prime|\)
        \State Find \((i^*,j^*) = \arg \max_{i,j} \ \residual{i}{j}\)
        \ForAll{\(c_i \in \cC(v_j^*)\)}
        \State Set \(\cv{i}{j^*} = \cv{i}{j^*}^\prime\) \ \text{and} \ \(\residual{i}{j^*} = 0\)
        \State Update \(\beta_{i,\muspace j^*} = \beta_{i,\muspace j^*}\times\beta\)
        \EndFor
        \State Update \(e = e + 1\)
        \ForAll{\(c_i \in \cC(v_j^*)\)}
        \State Generate and propagate \(\vc{j^*}{i}\)
        \ForAll{\(v_j \in \cV(c_i)\setminus \{v_j^*\}\)}
        \State Compute and store \(\cv{i}{j}^\prime\)
        \State Compute \(\residual{i}{j}\) according to \eqref{decay}
        \EndFor
        \EndFor
        \If{\(e\) equals \(E-N\)}
        \State Reset all \(\beta_{i,\muspace j} = 1\) and \(e = 0\)
        \EndIf
        \State \textbf{if} stopping rule is not satisfied \textbf{then}
        \State \hspace{0.5cm} Continue execution from line 5;
        \State \textbf{end if}
    \end{algorithmic}
    \label{alg1}
\end{algorithm}

The message-update strategy employed by C\&R-RBP, as illustrated in Figure~\ref{Fig:CR-RBP}, reveals its inherently \emph{variable-node-centric} nature.
In particular, \(v_{j^*}\) can be interpreted as the least reliable VN at a given stage of the decoding process, and all C2V and V2C messages associated with \(v_{j^*}\) are updated accordingly.
This contrasts with \emph{check-node-centric} approaches, such as NW-RBP \cite{Casado2010}, in which the scheduling process is organized around selected CNs.

\subsection{Proposed C\&DR-RBP Algorithm}\label{C-VN}

Motivated by empirical observations indicating a significant improvement in error-rate performance when the \emph{Return} strategy is applied only after a few iterations, we devised the C\&DR-RBP algorithm summarized in Algorithm~\ref{alg2}.
The modification in line 11 suppresses the update of the returning message, whereas the modification in line 18 preserves the same computational complexity in terms of residual calculation, as we discuss in the next section.

\begin{algorithm}[ht]
    \caption{C\&DR-RBP decoding.}
    \begin{algorithmic}[1]
        \State Given \(i_0\), for iterations \(i < i_0\), apply Algorithm~\ref{alg1} changing the following lines:
        
        11. \textbf{for all} \(c_a \in \cC(v_j^*)\setminus\{c_{i^*}\}\) \textbf{do}
        
        18. \textbf{if} \(e\) equals \(E\) \textbf{then}
        \State For \(i \ge i_0\), apply Algorithm~\ref{alg1} unchanged.
    \end{algorithmic}
    \label{alg2}
\end{algorithm}

The proposed decoding algorithms could be used in iterative detection and decoding and message passing applications \cite{spa,mfsic,mbdf,bfidd,1bitidd,did,baplnc,dynovs,listmtc,msgamp,3lmp,iddllr,comp,cdidd,oclidd}, where they can reduce the required number of outer and inner iterations.
	
	\section{Analysis}\label{Analysis}
	
        In this section, we analyze the overall computational complexity of the proposed algorithms in terms of the average number of C2V, V2C and residual calculations, and also real-value comparisons.
We also develop a theoretical analysis showing that the \emph{Consensus} strategy increases the probability of a correct bit decision under the initial conditions of the decoding process.

\subsection{Computational Complexity}\label{Complexity}

\begin{table*}[ht]
    \centering\scriptsize
    \caption{Overall computational complexity per algorithm iteration in terms of number of calculations of C2V, V2C, residuals and real-value comparisons.}
    \begin{tabular}{@{}c c c c c c @{}}
        \toprule
        Algorithm & \scriptsize Algorithm loops & C2V & \scriptsize V2C & \scriptsize Residuals & \scriptsize Real value comparisons \\
        \midrule
        LLR-BP \cite{MacKay1999}, LBP \cite{Hocevar2004} & \(E\) & \(E\) & \(E\) & 0 & 0\\
        RBP \cite{Casado2010}, RD-RBP \cite{Zhang2019}, SVNF \cite{Lee2013}, UBP-RBP \cite{Chang2022} & \(E\) & \(E\) & \(E(\bard_v - 1)\) & \(E(\bard_v-1)(\bard_c-1)\) & \(E(E-1)\) \\	
        NW-RBP \cite{Casado2010} & \(E\) & \(E\)& \(E(\bard_v - 1)\) & \(E(\bard_v-1)(\bard_c-1)\) & \(M(E-1)\) \\
        List-RBP \cite{Yan2024} & \(E\)  & \(E\)& \(E(\bard_v - 1)\) & \(E(\bard_v-1)(\bard_c-1)\) &  \( E (\bard_v-1)(\bard_c-1)S_2 + S_1 \muspace S_2\) \\
        C\&R-RBP, C\&DR-RBP (\(i\ge i_{\scaleto{0}{2pt}}\))& \(E - N\) & 0 & \(E(\bard_v - 1)\) & \(E(\bard_v-1)(\bard_c-1)\) & \((E-N)(E-1)\) \\
        C\&DR-RBP (\(i\ \!\!<\! i_{\scaleto{0}{2pt}}\))  & \(E\) & 0 & \(E(\bard_v - 1)\) & \(E(\bard_v-1)(\bard_c-1)\) &  \(E(E-1)\)  \\
        \bottomrule
    \end{tabular}
    \label{tab:complexity}
\end{table*}

The computational bottleneck of residual-based IDS algorithms is, to a great extent, the number of residual calculations.
Therefore, to ensure a fair comparison among different methods, we set the expected number of residual calculations, denoted by \(\rho\), to be identical for both C\&R-RBP and RBP\@.
The total average number of residual calculations \(\rho\) is equal to the expected total number \(C\) of C2V precomputations.
Thus,
\begin{equation}\label{rho}
    \rho = C \Rightarrow \rho =  A\cdot\left(\frac{V}{A}\right)\cdot\left(\frac{C}{V}\right),
\end{equation}
where \(A\) and \(V\) denote, respectively, the expected number of algorithm loops and V2C updates.
Let \(\bard_c\) and \(\bard_v\) denote the average CN and VN degrees of the Tanner graph such that 
\[
    E = \bard_c \,  M = \bard_v \, N.
\]
From the discussion in Subsection~\ref{RBP}, for each updated V2C message, there are \(\bard_c-1\) C2V precomputations to evaluate residuals (see Figure~\ref{Fig:RBP}).
Hence,
\begin{equation}\label{1}
    \frac{C}{V} = \bard_c-1.
\end{equation}
For each algorithm loop, there are \(\bard_v-1\) V2C updates.
Thus,
\begin{equation}\label{2}
    \frac{V}{A} = \bard_v-1.
\end{equation}
Thus, since \(A = E\) for RBP, from \eqref{rho}, \eqref{1} and \eqref{2}, we have
\begin{equation}\label{rho2}
    \rho = E(\bard_v-1)(\bard_c-1).
\end{equation}
Note that C\&R-RBP requires one more V2C update per algorithm loop than RBP due to \emph{Return}; i.e.,
\begin{equation*}
    \frac{V}{A} = (\bard_v - 1) + 1 = \bard_v\,.
\end{equation*}
Accordingly, we set \(A = E-N\), which results in an average number of V2C updates equal to
\begin{IEEEeqnarray*}{rCl}
    V & = & A\cdot\left(\frac{V}{A}\right)=(E-N)\bard_v = (E-N)\frac{E}{N} = \left(\frac{E}{N}-1\right)\!E \\
    & = & E(\bard_v-1).
\end{IEEEeqnarray*}
Since \(C \slash V= \bard_c-1\) for C\&R-RBP, the number of residual calculations matches that of classical RBP, as shown in \eqref{rho2}.

For C\&DR-RBP, before the \emph{Return} strategy is applied, the number of residual calculations remains identical to that of RBP because the returning message is not updated.
Therefore, if \emph{Return} is applied at iteration \(i_0\), for \(i < i_0\) we set \(A = E\) to preserve the complexity \(\rho\) given in \eqref{rho2}.

Table~\ref{tab:complexity} summarizes the overall computational complexity of the proposed algorithms, LLR-RBP (flooding), LBP, and other RBP-based methods with similar overall computational complexity.
In general, the number of C2V calculations is equal to the number of algorithm loops, except for the proposed methods, for which it is zero because the messages \(\cv{i}{j} ^\prime\), computed during residual evaluation (lines~3 and 14 of Algorithm~\ref{alg1}), are stored in memory for later reuse.
We also indicate the \emph{theoretical} number of real-value comparisons required to find the largest residual, which is generally equal to \(E\) times the number of comparisons required per algorithm loop.
Since there are \(E\) residuals, \(E-1\) comparisons are required, resulting in a expected total of \(E(E-1)\) comparisons\footnote{See Appendix~\ref{appD} for an estimation of real-value comparisons of List-RBP.}.
However, this theoretical value can be significantly reduced in practical implementations for all RBP-based algorithms, as we discuss in Subsection~\ref{Nops}.

\subsection{Analysis of Consensus under Initial Conditions}\label{Math}

In this subsection, we show that \emph{Consensus} strategy increases the probability of correct decoding under initial decoding conditions and in the presence of severe bit transmission errors.
This analysis provides insight into the performance gains observed in the simulation results.

\subsubsection{Severe bit errors}

RBP-based algorithms exhibit greedy behavior because message updates are guided by the largest C2V residual criterion.
Under initial decoding conditions, i.e., at the start of the decoding process, we have, from \eqref{LLR}, \eqref{v2c}, and \eqref{decay},
\begin{IEEEeqnarray*}{rCl}
    \vc{j}{i} & = & \ell_j = \tfrac{2\abs{y_k}}{\sigma^2}, \\
    \residual{i}{j} & = & | \,\cv{i}{j}^\prime |.
\end{IEEEeqnarray*}
Therefore, the largest residual criterion selects the C2V message with the highest magnitude.
In this sense, the algorithm tends to prioritize C2V messages that are more likely to be correct.
To further examine this, consider the \emph{min-sum} approximation of the C2V message computation \cite[p.~226]{Ryan2009}:
\begin{equation}\label{minsum}
    \cv{i}{j} \approx \bigg(\prod_{v_k \in \cV(c_i)\setminus \{v_j\}} \!\!\!\!\!\! \alpha_{v_k}\bigg) \min_{v_k \in \, \cV(c_i)\setminus \{v_j\}} \abs{\vc{k}{i}},
\end{equation}
where \(\alpha_{v_k} = \text{sign}(\vc{k}{i})\).
An LDPC decoder that employs \eqref{minsum} to compute C2V messages is referred to as a \emph{min-sum decoder}.
Under this approximation, the largest residual criterion tends to select for updating, under initial decoding conditions,
\begin{IEEEeqnarray*}{rCl}
    (i^*\!,\muspace j^*) & = & \arg \max_{i,\, j} \min_{v_k \in \, \cV(c_i)\setminus \{v_j\}} \abs{\vc{k}{i}} \\
    & = & \arg \max_{i,\, j} \min_{v_k \in \, \cV(c_i)\setminus \{v_j\}} \tfrac{2\abs{y_k}}{\sigma^2}. \IEEEyesnumber \label{type1}
\end{IEEEeqnarray*}

To better understand the mechanism described in \eqref{type1}, consider the sets of the form \(\cV(c_i)\setminus \{v_j\}\), i.e., the sets of VNs connected to the same CN \(c_i\), \emph{excluding} one of those VNs.
For convenience, we refer to such sets as a \emph{type-1 VN set}.
Figure~\ref{fig:min-sum} illustrates an example of a type-1 VN set, represented by the VNs \(\{v_1, v_2, v_3\}\) (blue circles).
Each C2V message is associated with a distinct type-1 VN set.

\begin{figure}[ht]
    \centering
    \includegraphics[width=0.75\linewidth]{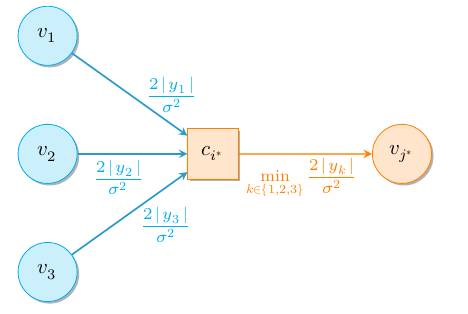}
    \caption{An example of a type-1 VN set (blue circles) under initial decoding conditions. The message \(\cvopt\) is calculated using min-sum.}
    \label{fig:min-sum}
\end{figure}

Now consider, for each type-1 VN set, the V2C message with the smallest magnitude.
If the message \(\cvopt\) in Figure~\ref{fig:min-sum} is selected for updating, then \(\{v_1, v_2, v_3\}\) corresponds to the type-1 VN set whose smallest V2C message is larger than that of all other type-1 VN sets, according to the criterion in \eqref{type1}.
In this sense, the set \(\{v_1, v_2, v_3\}\) can be interpreted as the most reliable VN subset in the Tanner graph.
Consequently, the RBP mechanism tends to select the C2V message associated with the most reliable type-1 VN set.

Therefore, it is intuitive that this mechanism is likely to select the type-1 VN set in which all observations \(y_j\) received are reliable.
Consequently, this mechanism selects a type-1 VN set containing at least one flipped observation (i.e., where the sign of \(y_j\) differs from that of transmitted signal \(u_j\) in \eqref{yj}) only in the presence of a severe bit error, in which the received observation \(y_j\) is not only incorrect in sign, but also has a large magnitude.  
As a result, although the RBP mechanism increases the likelihood of selecting the most reliable type-1 VN set, it also introduces the risk of selecting type-1 VN sets with severely corrupted observations.
Once such erroneous information is propagated through the Tanner Graph, correcting it becomes considerably more difficult due to the large error magnitude.
This interpretation, developed herein using the min-sum approximation, provides useful insight into the origin of the performance degradation associated with the greedy behavior of RBP  and helps explain why mitigating this effect is particularly challenging.

\subsubsection{Probability of correct decision given  severe bit error}

The \emph{Consensus} strategy increases the probability of a correct bit decision in the presence of a severe bit error under the initial conditions of the decoding process.
Consequently, it mitigates the performance degradation associated with the greedy behavior discussed above.
The following analysis explains this effect.

Let \(M_{c_a \rightarrow v_j}\) and \(L_j\) denote the random variables whose observations are \(\cv{a}{j}\) and \(\ell_j\muspace\), respectively.
Suppose that \(\cvopt^\prime\) is the C2V message with the largest magnitude; thus, by the largest residual criterion, we assign \(\cvopt = \cvopt^\prime\).
Let
\begin{equation}\label{Delta}
    \Delta = \sum_{c_a \in \, \cC(v_j^*)\setminus\{c_i^*\}} M_{c_a \rightarrow v_j^*}
\end{equation}
denote the sum of the additional C2V messages updated by the \emph{Consensus} strategy.
Therefore, the random variable corresponding to the posterior LLR in \eqref{post} can be written as
\begin{equation*}
    L_j^{\prime*} = L_j^* + \cvopt + \Delta\muspace.
\end{equation*}

The proposed method will outperform RBP under initial conditions if it yields a higher probability of correct bit decisions, according to the test in \eqref{estimate}.
Since \(\Delta = 0\)  for RBP under initial decoding conditions, we seek to show, for VN \(v_j^*\muspace\), that
\begin{equation}\label{Prob}
    \Prob{L_j^* + \cvopt + \Delta < 0} \ge \Prob{L_j^* + \cvopt \!< 0}.
\end{equation}
From \eqref{bitmapping}, \eqref{yj}, \eqref{LLR} and under the AWGN assumption, \(L_{j}\) is Gaussian-distributed and satisfies
\begin{equation}\label{ELj}
    \E{L_j} =\frac{2(-1)^{b_j}}{\sigma^2}.
\end{equation}
Assume that \(\Delta\) is Gaussian-distributed and independent of \(L_{j}^*\muspace\).
According to \cite{Chung2001}, this assumption is empirically justified for message-passing algorithms when the number of iterations is small, which is consistent with the initial decoding conditions considered here.
Therefore, from \eqref{Prob}, we obtain
\begin{equation}\label{cdf}
    \Phi\!\left(\tfrac{-\E{L_j^* - \Delta} -\, \cvopt}{\sigma_{L_j^*} + \,\sigma_{\Delta}}\right) \ge \Phi\!\left(\tfrac{-\E{L_j^*} - \,\cvopt}{\sigma_{L_j^*}}\right)\!,
\end{equation}
where \(\Phi(z) = \Prob{Z \le z}\) is the cumulative density function (CDF) of the standard normal random variable \(Z\sim\cN(0,1)\).
Since \(\Phi(\cdot)\) is a monotonically increasing function, the inequality in \eqref{cdf} is equivalent to
\begin{equation}
    \frac{\E{L_j^*} + \,\E{\Delta} + \,\cvopt}{\sigma_{L_j^*} + \,\sigma_{\Delta}} \le \frac{\E{L_j^*} +\, \cvopt}{\sigma_{L_j^*}},
\end{equation}
which can be reduced to
\begin{equation*}
    \E{\Delta} \le \frac{\sigma_{\Delta}}{\sigma_{L_j ^*}}\left(\E{L_j^*} + \cvopt\right).
\end{equation*}
Let us assume, without loss of generality, that \(b_j^* = 1\) was transmitted.
From \eqref{ELj}, we obtain
\begin{equation}\label{inequality}
    \E{\Delta} \le \frac{\sigma_{\Delta}}{\sigma_{L_j ^*}}\left(\cvopt - \tfrac{2}{\sigma^2}\right).
\end{equation}

The inequality in \eqref{inequality} is equivalent to the condition in \eqref{Prob}.
Note that, since \(\sigma_{\Delta}/\sigma_{L_j ^*} > 0\), the condition
\begin{equation}\label{extreme}
    \cvopt > 2/\sigma^2
\end{equation}
is sufficient for \eqref{inequality} to hold, provided that \(\E{\Delta} < 0\).
As shown later, condition \eqref{extreme} corresponds to a severe bit error event when \(b_j^* = 1\).
Consequently, if we can show that \(\E{\Delta} < 0\), then whenever \eqref{extreme} holds, condition \eqref{inequality} also holds.
Therefore, we will have proved that \emph{Consensus} improves the probability of correct bit decisions under the initial conditions and in the presence of severe bit errors.

\begin{proposition}\label{Theorem}
    Let  \(g: \sR^n \rightarrow \sR\) be such that 
    \begin{equation*}
        x_j \ge 0 \ \forall j \in \{1, ... \,,n\} \Rightarrow g(x_1, ... \,, x_n) \ge 0.
    \end{equation*}
    Furthermore, let \(g(\cdot)\) be an odd function with respect to each of its arguments; i.e.,
    \begin{equation*}
        g(...\,, -x_j\muspace, ...) = -g(...\,, x_j\muspace, ...), \ j \in \{1,...\,,n\}.
    \end{equation*}	
    Let \(\{L_j\}_{j=1}^n\) be a sequence of independent random variables with marginal probability density functions \(\{\,f_j(\ell_j)\}_{j=1}^N\).
    Call \(M = g(L_1, ... \,, L_n)\) and define
    \[
    h_j(\ell_j) = f_j(\ell_j)-f_j(-\ell_j).
    \]
    If the expectation \(\E{M}\) exists, it is true that
    \begin{equation*}
        \E{M} = \intzeroinf ... \,\intzeroinf g(\ell_1,... \,,\ell_n)\prod_{j=1}^{n}h_j(\ell_j)d\ell_1... \, d\ell_n.
    \end{equation*}
\end{proposition}

\begin{proof}
    See Appendix~\ref{appA}.
\end{proof}

\begin{corollary}\label{Corollary}
    If each \(\{L_j\}_{j=1}^n\) in Theorem~\ref{Theorem} are Gaussian-distributed, then 
    \begin{equation*}
        s(\E{M}) = \prod_{j=1}^n s(\E{L_j}),
    \end{equation*}
    where \(s(x)\) denotes the sign of \(x \in \sR\,\).
\end{corollary}

\begin{proof}
    See Appendix~\ref{appB}.
\end{proof}

Since both \(\tanh(x)\) and \(\tanh^{-1}(x)\) are odd non-decreasing functions, one can verify that the C2V update rule defined in \eqref{c2v} satisfies the required properties in Theorem~\ref{Theorem}.
Furthermore, since the \(L_{j}\) are independent under the AWGN assumption (see \eqref{yj} and \eqref{LLR}), it follows from Theorem~\ref{Theorem} and Corollary~\ref{Corollary} that
\begin{equation*}
    s\!\left(\E{M_{c_i \rightarrow v_{j}}}\right) = \prod_{k:v_k\in\cV(c_i)\setminus \{v_j\}} s\left(E\left(L_{k}\right)\right).
\end{equation*}
According to \eqref{ELj}, \(s\!\left(\E{L_j}\right) = (-1)^{b_j}\); thus, we have
\begin{equation}\label{signM}
    s\!\left(\E{M_{c_i \rightarrow v_{j}}}\right) = \prod_{k:v_k\in\cV(c_i)\setminus \{v_j\}}(-1)^{b_k}.
\end{equation}
Because each CN \(c_i\) represents a single parity-check equation, and according to the bit mapping in \eqref{bitmapping}, we must have 
\begin{equation}\label{parity}
    \prod_{k:v_k\in\cV(c_i)}(-1)^{b_k} = 1, \forall c_i\muspace.
\end{equation}
Thus, multiplying \eqref{signM} by \((-1)^{2b_{j}}\!=\!1\) and from \eqref{parity}, we have
\begin{IEEEeqnarray*}{rCl}
    s\!\left(\E{M_{c_i \rightarrow v_{j}}}\right)
    & = & \prod_{k:v_k\in\cV(c_i)\setminus \{v_j\}}(-1)^{b_k} \cdot (-1)^{b_{j}}\cdot(-1)^{b_{j}} \\
    & = & \prod_{k:v_k\in\cV(c_i)}(-1)^{b_k}\cdot(-1)^{b_{j}} \\ 
    & = & (-1)^{b_{j}}. \IEEEyesnumber \label{signM2}
\end{IEEEeqnarray*}

Since we assumed a transmitted bit \(b_{j}^* = 1\), for each \(c_a \in \, \cC(v_j^*)\setminus\{c_i^*\}\), it follows from \eqref{signM2} that
\begin{equation}\label{sEM}
    s\left(\E{M_{c_a \rightarrow v_j^*}}\right) = -1 < 0.
\end{equation}
Thus, from the definition of \(\Delta\) in \eqref{Delta}, we have shown that \(\E{\Delta} < 0\).
Consequently, \eqref{inequality} holds whenever \eqref{extreme} is true.
It remains to show that the condition in \eqref{extreme} corresponds to a severe error event for \(b_j^* = 1\).
First, we use the alternative definition of C2V message update presented in \cite[p.~221]{Ryan2009}:
\begin{equation}\label{c2valt}
    \cv{i}{j} = \bigg(\prod_{v_k \in \cV(c_i)\setminus \{v_j\}} \!\!\!\!\!\! \alpha_{v_k}\bigg) \phi\bigg( \sum_{v_k \in \cV(c_i)\setminus\{v_j\}} \! \phi(\beta_{v_k})\bigg),
\end{equation}
where \(\alpha_{v_k} = s(\vc{k}{i})\), \(\beta_{v_k} = |\vc{k}{i}|\), and \(\phi\) is a function defined as
\begin{equation}\label{phi}
    \phi(x)=\ln\left(\frac{\e^x+1}{\e^x - 1}\right), \ x > 0.
\end{equation}
The following properties of \(\phi\) are straightforward to verify:
\begin{enumerate}
    \item[\(i\muspace\))] \(\phi(x) > 0\) for \(x >0\).
    \item[\(ii\muspace\))] \(\phi(x)\) is monotonically decreasing.
    \item[\(iii\muspace\))] \(\phi^{-1}(x) = \phi(x)\).
\end{enumerate}
The following theorem establishes an important property of C2V messages.

\begin{proposition}\label{Theorem2}
    It follows from the properties of the function \(\phi\) defined in \eqref{phi} that
    \begin{equation}\label{abs1}
        \abs{\cv{i}{j}} \le \min_{v_k \in \, \cV(c_i)\setminus \{v_j\}} \abs{\vc{k}{i}}.
    \end{equation}
\end{proposition}

\begin{proof}
    See Appendix~\ref{appC}.
\end{proof}

Under initial decoding conditions, the inequality in \eqref{abs1} becomes
\begin{equation}\label{abs}
    \abs{\cv{i}{j}} \le \min_{v_k \in \, \cV(c_i)\setminus \{v_j\}} \frac{2\abs{y_k}}{\sigma^2}.
\end{equation}
Therefore, since \(|\cvopt|\) is assumed to correspond to the largest residual, and since under initial decoding conditions \(\vc{k}{i^*} = 2y_k/\sigma^2\), it follows from \eqref{abs} that all \(\abs{y_k}\) in the type-1 VN set \(\cV(c_{i^*})\setminus \{v_{j^*}\}\) must be sufficiently large; i.e., the magnitudes must be sufficiently distant from zero.
Furthermore, from \eqref{c2valt}, we have
\begin{equation}\label{sign}
    s(\cvopt) = \prod_{v_k \in \, \cV(c_i^*)\setminus \{v_{j^*}\!\}} s\!\left(\vc{k}{i^*}\right).
\end{equation}
Consequently, due to \eqref{sEM} and under the assumption that \(b_j^* = 1\), an observed \(\cvopt > 0\) indicates the existence of at least one VN \(v_{j_0}\) in the type-1 VN set \(\cV(c_{i^*})\setminus \{v_{j^*}\}\) whose received observation has a large magnitude \(\abs{y_{j_0}}\) and an incorrect sign relative to the bit mapping in \eqref{bitmapping}; i.e., \(s(y_{j_0}) \ne s(u_{j_0})\).
Assuming that \eqref{extreme} holds, it follows from \eqref{abs} that
\begin{equation*}
    \frac{2}{\sigma^2} < \cvopt \le \frac{2\abs{y_{j_0}}}{\sigma^2};
\end{equation*}
thus,
\begin{equation*}
    \abs{y_{j_0}} \ge 1.
\end{equation*}
Equivalently,
\begin{equation}
    \begin{cases}
        y_{j_0}  > \phantom{-}1, \text{ if } b_{j_0} = 1, \\
        y_{j_0}  < -1, \text{ if } b_{j_0} = 0,
    \end{cases}
\end{equation}
which, according to the bit mapping in \eqref{bitmapping}, corresponds to severe bit transmission error at VN \(v_{j_0}\muspace\).
Therefore, condition \eqref{extreme} characterizes the presence of a severe transmission error.

Since we have shown that \(\E{\Delta} < 0\) for \(b_j^* = 1\), we have also shown that inequality \eqref{inequality} holds.
Consequently, under initial decoding conditions, the \emph{Consensus} strategy increases the probability of a correct bit decision in the presence of severe bit transmission errors.
	
	\section{Considerations Regarding 5G NR codes}\label{considerations}
	
        In this work, we use the 5G NR coding standard in our simulations because of its strong practical relevance as one of the fastest-growing communication standards today.
Therefore, before presenting the simulation results, it is important to discuss some aspects of 5G NR LDPC codes related to latency and computational complexity in the context of informed dynamic scheduling (IDS) methods.

In particular, previous works such as \cite{Chang2022} and \cite{Wang2020} have identified structural features of 5G NR LDPC codes that can affect the efficiency of IDS-based decoders.
\cite{Chang2022} argues that the presence of degree-1 VNs in the 5G NR Tanner graph leads to a disproportional selection of these VNs when IDS strategies are employed.
As a result, C2V messages associated with degree-1 VNs are more likely to be selected and updated than other C2V messages connected to the same CN\@.
However, allocating decoding resources to repeatedly update degree-1 VNs reduces decoding efficiency, since these nodes can only contribute their intrinsic channel information to the connected CN\@.
To mitigate this effect, \cite{Chang2022} introduce the degree-1 VN (D1VN) processing scheme, which limits the updates associated with degree-1 VNs.

Another issue related to 5G NR LDPC codes was discussed in \cite{Wang2020}, namely the impact of punctured VNs on the code graph.
Since the decoder pads zeros at the punctured locations, if a CN \(c_i\) has two punctured VNs, all messages transmitted from \(c_i\) to its neighboring VNs become zero.
Furthermore, if a CN is connected to only one punctured VN, the corresponding punctured symbol can be recovered after a single update.
Based on this observation, \cite{Wang2020} proposed prioritizing the update of such CNs through a scheduling strategy called least-punctured (LP) scheduling.

These observations are important for practical implementations of 5G NR systems.
However, since the primary objective of this work is to compare new RBP-based scheduling strategies with existing RBP-like methods under a common framework,
we adopt standard decoding settings without incorporating additional mechanisms  specifically designed for particular code structures, such as those discussed above.
This ensures a fair and consistent comparison among algorithms.

    \section{Implementation and Parallelization}\label{Implementation}

        Several works have investigated different implementations of LDPC codes for real-world applications \cite{Wang2007a,LeGal2016,Hailes2016,Wu2019a,Xu2019,Behera2020,Liu2020,Yang2021,Tran-Thi2023,Benhayoun2023,Sy2023,Yun2025}.
As discussed in \cite{Benhayoun2023}, two broad approaches can be considered for LDPC decoder implementations: hardware-based and software-based implementations.
Hardware implementations based on application-specific integrated circuits (ASICs) or field-programmable gate arrays (FPGAs) can achieve low latency and high throughput, but generally require greater development effort and hardware resources.
In a survey of LDPC decoder implementations on FPGAs, the authors of \cite{Hailes2016} point out that considerable effort has been devoted to designs that exploit the flexibility, high processing speed, and parallelism of FPGAs.
These characteristics make \mbox{FPGA}s particularly attractive for design prototyping and for the manufacturing of small production-run devices.
Their in-system programmability can also make them more cost-effective than ASICs in such applications.
Examples of FPGA implementations can be found in \cite{Behera2020,Liu2020,Yang2021,Tran-Thi2023}.

The second approach discussed in \cite{Benhayoun2023} is software-based implementation, in which the decoding algorithm is implemented in a programming language, compiled into a binary program, and executed on a general-purpose processor (GPP).
Software implementations can reduce the hardware resources and development time required for deployment and allow functionalities to be modified or updated without changing the underlying hardware.
Examples of software-based LDPC decoder implementations are presented in \cite{LeGal2016,Xu2019,Benhayoun2023,Sy2023,Yun2025}.
According to \cite{Sy2023}, the development of software targets such as graphics processing units (GPUs) and general-purpose processors has enabled programmable, reconfigurable and scalable software-based LDPC decoders.
Moreover, the integration of vector units and instruction-set extensions in modern processors can significantly reduce execution time and improve the efficiency of parallel processing.
For example, \cite{Xu2019} exploits the AVX-512 instruction set to accelerate LDPC decoding.

\subsection{Parallelization and IDS methods}

Due to its parallel nature, most works on the implementation of LDPC decoders have focused on the original flooding scheme (see \cite{LeGal2016}).
Others have exploited the semi-parallel features of quasi-cyclic LDPC codes for efficient implementation of layered decoding processing schemes, most prominently the LBP algorithm \cite{Wang2007a,Liu2020,Yang2021,Tran-Thi2023}.

Given the dynamic nature of IDS methods, such as the RBP-based algorithms considered in this work, a full or partial parallelization of these methods is considerably challenging.
As pointed out by the authors of the original RBP algorithm in \cite{Casado2010}, the principle of performing an ordered sequence of updates while exploiting the most recent information as much as possible is not, in principle, compatible with simultaneously computing multiple messages.

Nevertheless, the serial nature of IDS methods can be relaxed: instead of updating only the largest residual in each algorithm loop, the \(p\) largest residuals could be updated in parallel.
This idea was suggested by the authors of the original RBP algorithms in \cite{Casado2010}.
More recently, the authors of \cite{Chang2021} followed this principle and proposed a parallel version of their CI-RBP.
This approach, however, generally leads to a degradation in performance, which tends to become more pronounced as \(p\) increases.
Therefore, there is a trade-off between parallelization and performance regarding IDS methods, which becomes particularly relevant when considering decoding latency: although parallelization can reduce the processing time per iteration, the associated performance degradation may require additional algorithm iterations to achieve the same error-rate performance, potentially offsetting part of the latency reduction.  
    
	\section{Simulation Results}\label{Simul}

\begin{figure}
    \centering
    \includegraphics[trim={2 0 8 0}, clip, width=\linewidth]{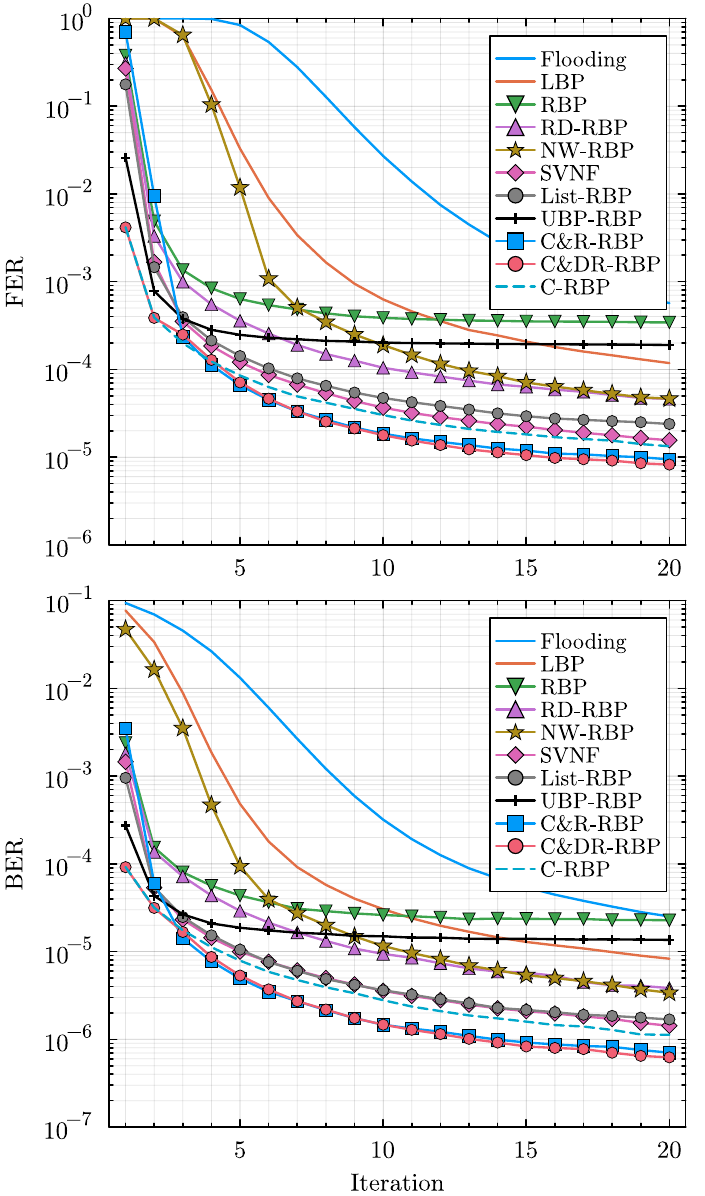}
    \caption{Simulation~1: FER/BER x Number of Iterations (\(E_b\!/N_0 = 3.0\)~dB), where a 5G NR \(576\)-code with \(R = 1/2\) is considered (\(\beta = 0.85\)).}
    \label{fig:sim1}
\end{figure}

\begin{figure}
    \centering
    \includegraphics[trim={8 0 8 0}, clip, width=\linewidth]{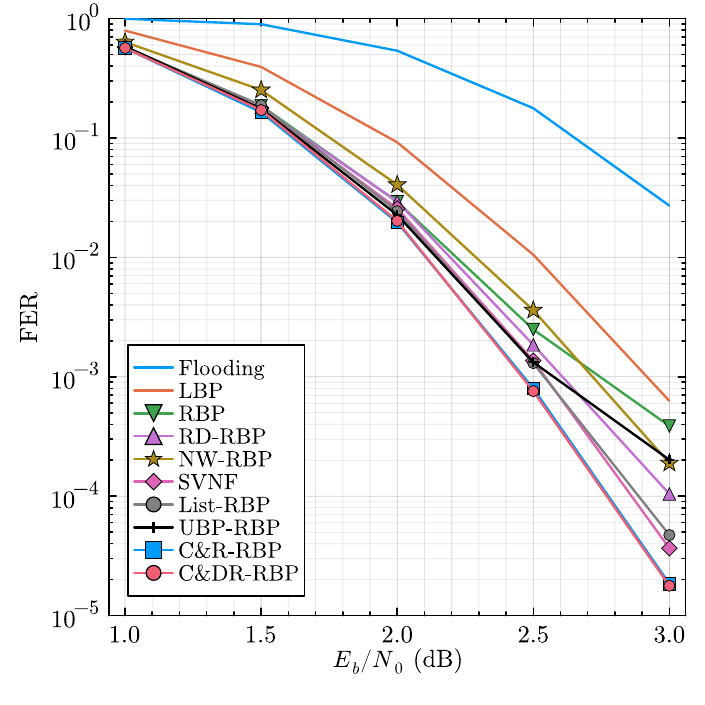}
    \caption{Simulation 1: FER x \(E_b\!/N_0\) for 10 algorithm iterations, where a 5G NR \(576\)-code with \(R = 1/2\) is considered (\(\beta = 0.85\)).}
    \label{fig:simdb}
\end{figure}

\begin{figure}
    \centering
    \includegraphics[trim={2 0 8 0}, clip,width=\linewidth]{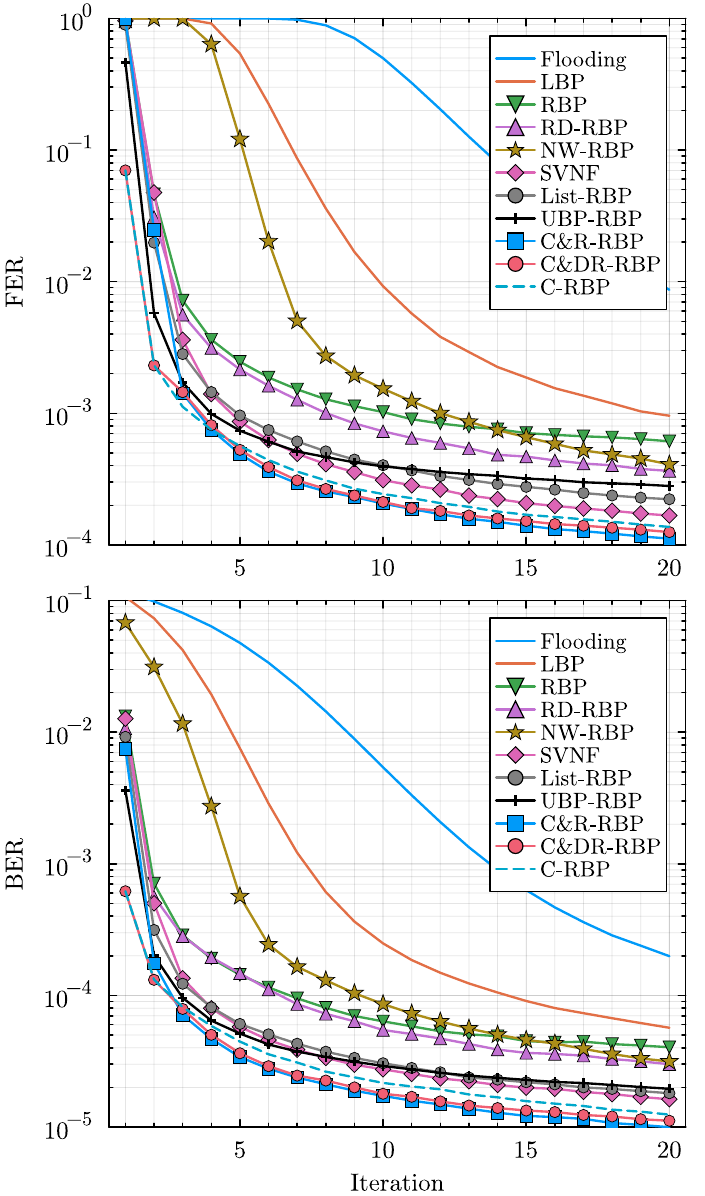}
    \caption{Simulation~2: FER/BER x Number of Iterations (\(E_b\!/N_0 = 2.0\)~dB), where a 5G NR \(1248\)-code with \(R = 1/2\) is considered (\(\beta = 0.85\)).}
    \label{fig:sim2}
\end{figure}

\begin{figure}
    \centering
    \includegraphics[trim={8 0 7 0}, clip,width=\linewidth]{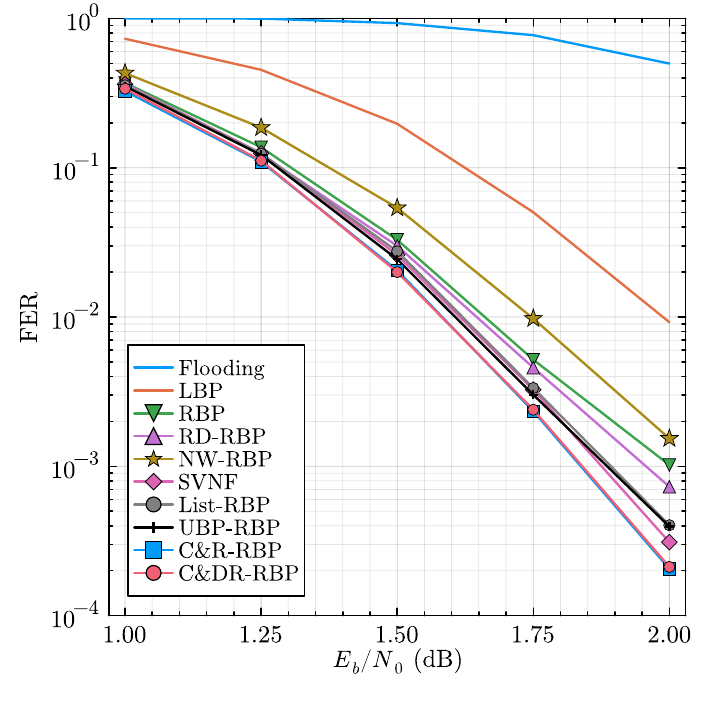}
    \caption{Simulation~2: FER x \(E_b\!/N_0\) for 10 algorithm iterations, where a 5G NR \(1248\)-code with \(R = 1/2\) is considered (\(\beta = 0.85\)).}
    \label{fig:simdb2}
\end{figure}


\begin{figure}
    \centering
    \includegraphics[trim={2 0 8 0}, clip,width=\linewidth]{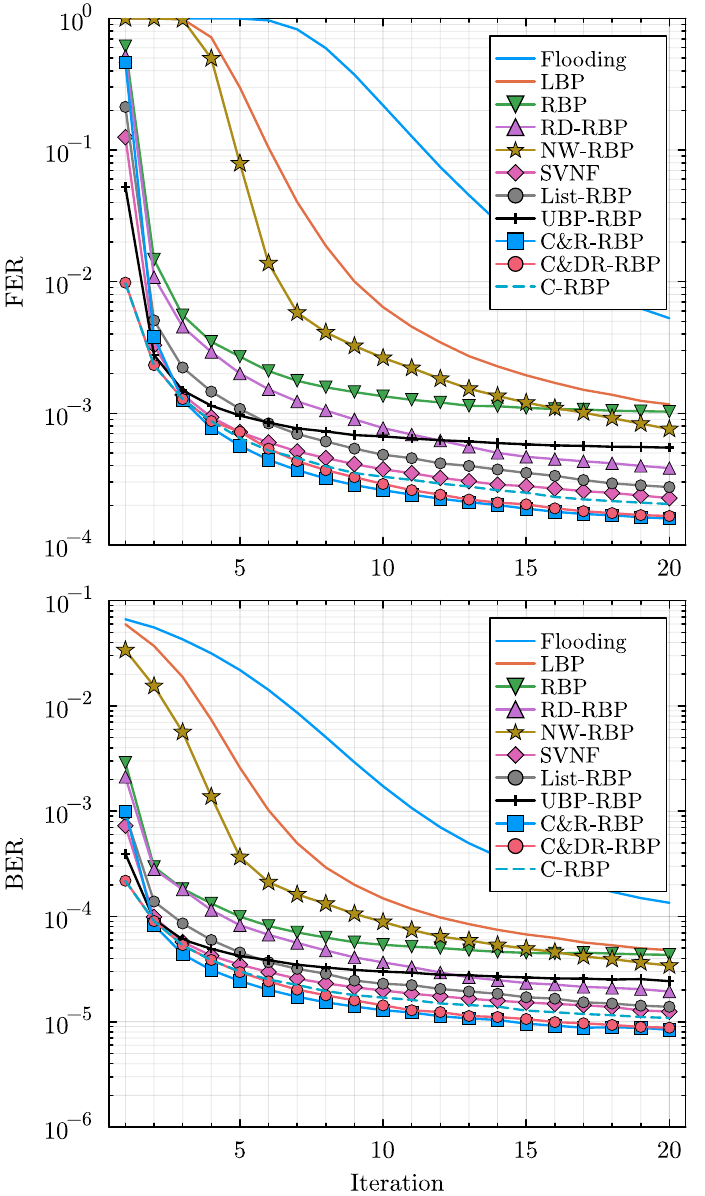}
    \caption{Simulation~3: FER/BER x Number of Iterations (\(E_b\!/N_0 = 2.8\)~dB), where a 5G NR \(1248\)-code with \(R = 2/3\) is considered (\(\beta = 0.85\)).}
    \label{fig:sim3}
\end{figure}

\begin{figure}
    \centering
    \includegraphics[trim={8 0 8 0}, clip,width=\linewidth]{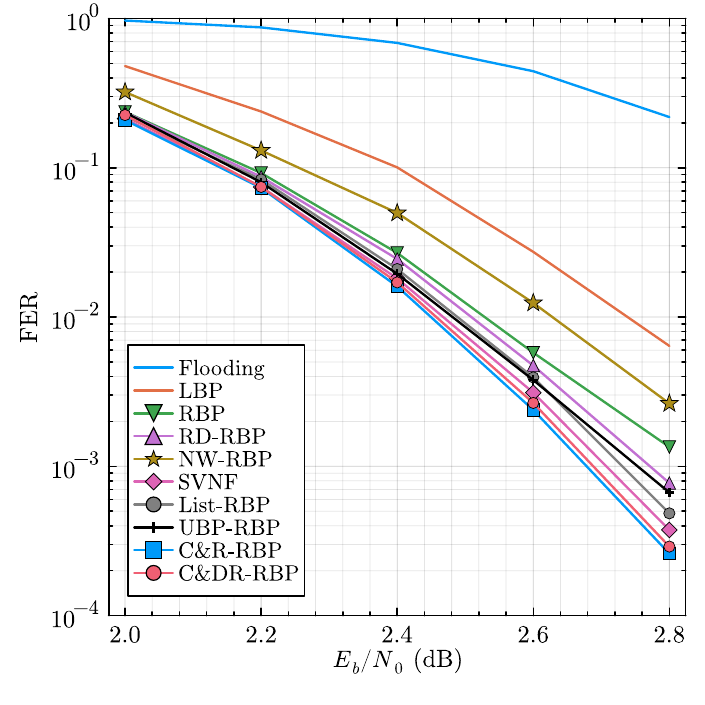}
    \caption{Simulation~3: FER x \(E_b\!/N_0\) for 10 algorithm iterations, where a 5G NR \(N = 1248\)-code with \(R = 2/3\) is considered (\(\beta = 0.85\)).}
    \label{fig:simdb3}
\end{figure}


\begin{figure}
    \centering
    \includegraphics[trim={2 0 8 0}, clip,width=\linewidth]{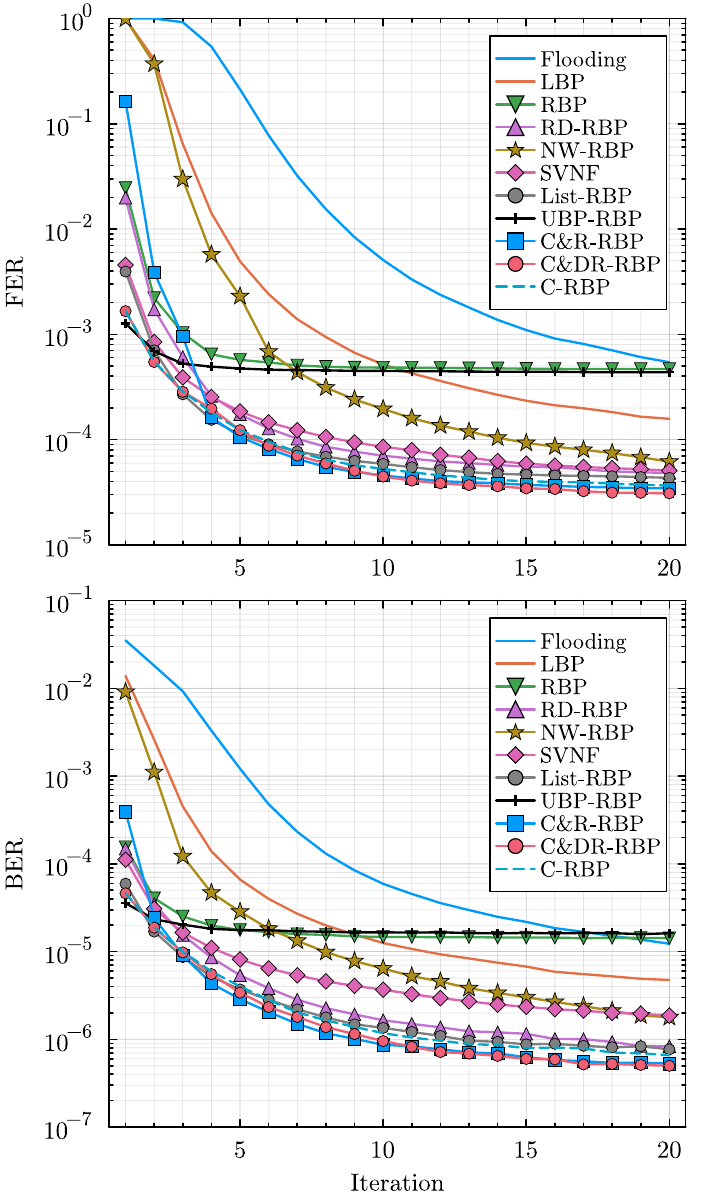}
    \caption{Simulation~4: FER/BER x No.~of Iterations (\(E_b\!/N_0 = 5.0\)~dB), where a 5G NR \(576\)-code with \(R = 5/6\) is considered (\(\beta = 0.85\)).}
    \label{fig:sim4}
\end{figure}

\begin{figure}
    \centering
    \includegraphics[trim={8 0 7 0}, clip,width=\linewidth]{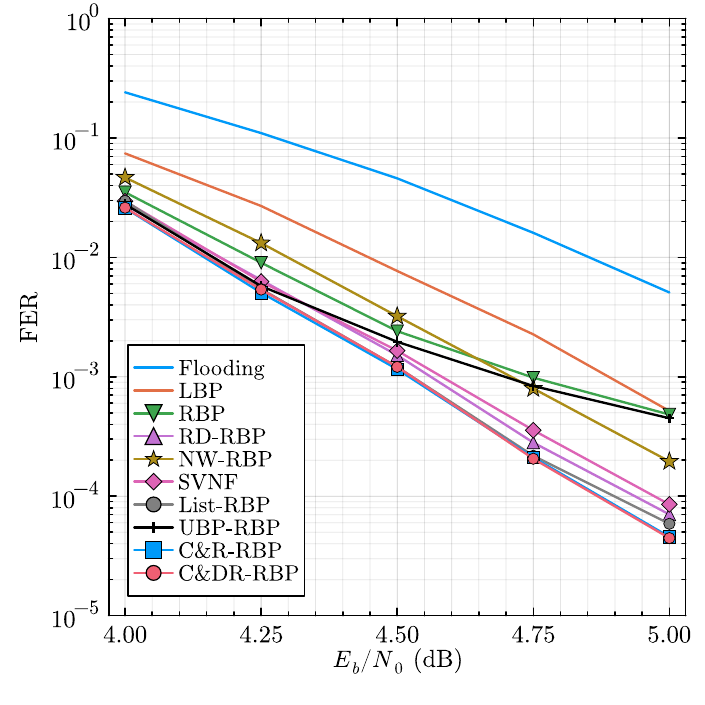}
    \caption{Simulation~4: FER x \(E_b/N_0\) for 10 algorithm iterations, where a 5G NR \(576\)-code with  \(R = 5/6\) is considered (\(\beta = 0.85\)).}
    \label{fig:simdb4}
\end{figure}


 \begin{figure}
    \centering
    \includegraphics[trim={2 0 8 0}, clip,width=\linewidth]{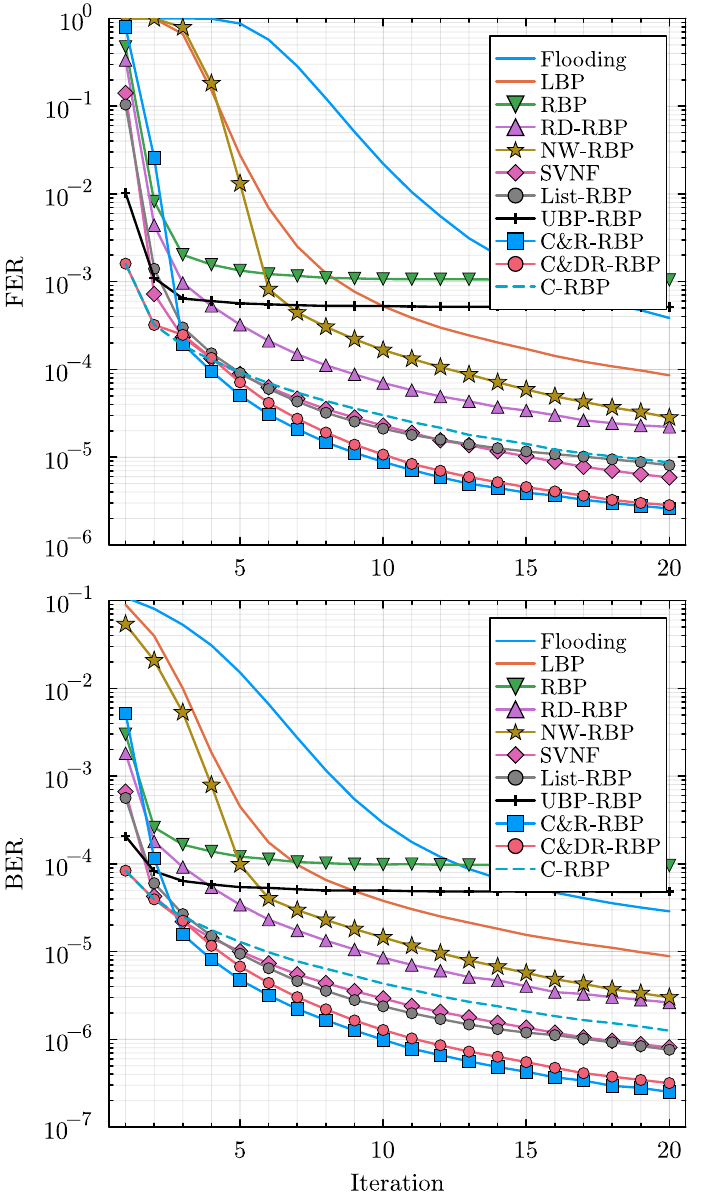}
    \caption{Simulation~5: FER/BER x No.~of Iterations under a Rayleigh fading channel (\(E_b\!/N_0 = 5.0\)~dB), where a 5G NR \(576\)-code with \(R = 1/2\) is considered (\(\beta = 0.85\)).}
    \label{fig:sim5}
\end{figure}

\begin{figure}
    \centering
    \includegraphics[trim={8 0 8 0}, clip,width=\linewidth]{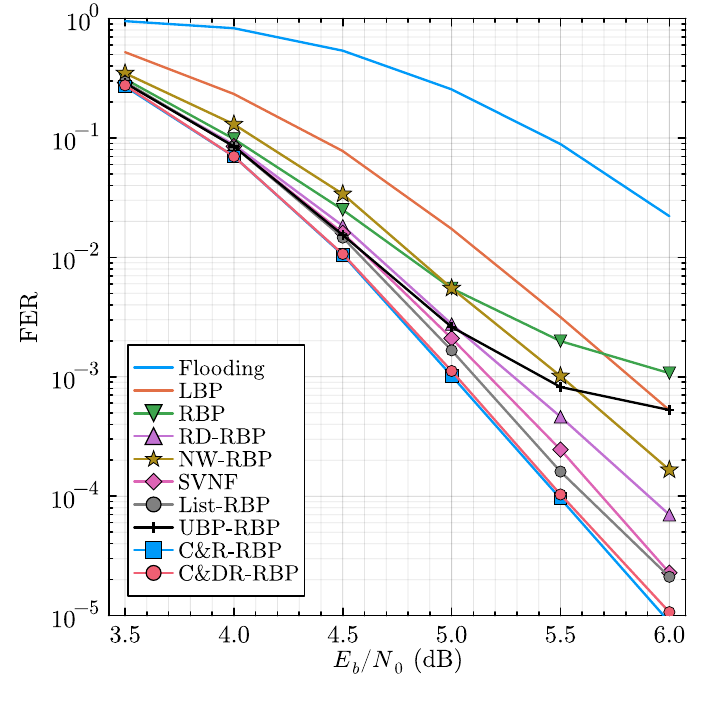}
    \caption{Simulation~5: FER x \(E_b/N_0\) under a Rayleigh fading channel model, for 10 algorithm iterations and where a 5G NR \(576\)-code with  \(R = 1/2\) is considered (\(\beta = 0.85\)).}
    \label{fig:simdb5}
\end{figure}

This section presents a performance evaluation of the proposed algorithms under different LDPC code configurations and channel conditions.
The proposed algorithms C\&R-RBP and C\&DR-RBP are compared with state-of-the-art decoding techniques in terms of frame error rate (FER) and bit error rate (BER).

\subsection{General Framework}

We evaluate the performance of the proposed algorithms through simulations using 5G NR LDPC codes.
The parity-check matrix \(\bmH\) is obtained by cyclic lifting from one of two possible base matrices, also referred to as base graphs (BG): (\textit{i}) BG1 for code lengths between 300 and 8448 and code rates between 1/3 and 8/9, and (\textit{ii}) BG2 for code lengths between 40 and 3840 and code rates between 1/5 and 2/3.
The lifting process is defined by the lifting number \(Z\) and the set index \(i_{LS}\muspace\).
The desired code rate is achieved (matched) by puncturing the least significant parity bits (see \cite{3GPP2025} for additional details).

We compare the performance of C\&R-RBP and C\&DR-RBP with RBP-based algorithm that have the same computational complexity in terms of residual calculations (see Table~\ref{tab:complexity}).
The performance curves presented in this section were obtained from Monte Carlo simulations using a sufficient number of trials to observe 252 frame errors at the last iteration (\(i = 20\)).
For the residual-decaying strategy, we follow \cite{Zhang2019} and \cite{Yan2024} and use a residual decaying factor \(\beta = 0.85\), which lies at the center of the recommended interval \((0.8,0.9)\), reported to provide the best performance.
This value is used for both proposed methods, as well for RD-RBP and List-RBP\footnote{For List-RBP, we empirically found the best performance with list sizes 16 and 2, which contrasts with the values suggested in \cite{Yan2024}.}.
We summarize the different code lengths, rates and \(E_b/N_0\) ranges for each simulation in Table~\ref{tab:simulations}.

\begin{table}[ht]
    \centering\footnotesize
    \caption{Code Lengths, rates and \(E_b/N_0\) ranges for each simulation.}
    \begin{tabular}{ccccc}
        \toprule
        Simulation & Code Length & Rate & \(E_b/N_0\) range & Figures \\
        \midrule
        1 & 576  & 1/2 & 1.0 : 3.0 & \ref{fig:sim1}, \ref{fig:simdb}\\
        2 & 1248 & 1/2 & 1.0 : 2.0 & \ref{fig:sim2}, \ref{fig:simdb2}\\
        3 & 1248 & 2/3 & 2.0 : 2.8 & \ref{fig:sim3}, \ref{fig:simdb3}\\
        4 & 576  & 5/6 & 4.0 : 5.0 & \ref{fig:sim4}, \ref{fig:simdb4}\\
        5 (fading) & 576  & 1/2 & 4.0 : 5.0 & \ref{fig:sim5}, \ref{fig:simdb5}\\
        \bottomrule
    \end{tabular}
    \label{tab:simulations}
\end{table}

\subsection{Simulation 1: Code Length = 576, R = 1/2}\label{Simul1}

Figure~\ref{fig:sim1} presents the resulting FER and BER for a rate-1/2 length-576 code at \(E_b/N_0 = 3.0\)~dB.
According to the 3GPP procedure described in \cite{3GPP2025}, this code configuration corresponds to base matrix BG2, lifting size \(Z = 40\) and set index \(i_{LS} = 2\).
For C\&DR-RBP, the \emph{Return} strategy is activated at iteration \(i_0 = 3\), chosen heuristically\footnote{See Appendix~\ref{appE} for a discussion on the sensitivity of C\&DR-RBP to the choice of \(i_0\).}.
As shown in Figure~\ref{fig:sim1}, the C\&R-RBP outperforms the competing algorithms for \(i \ge 3\) under both FER and BER metrics.
In contrast, C\&DR-RBP exhibits a noticeable performance improvement over C\&R-RBP during initial iterations and outperforms all competing algorithms at every iteration.
It is also important to note that C\&DR-RBP requires considerably fewer iterations than the other algorithms to achieve a given error-rate performance.
These results corroborate the conclusions of the formal analysis of the \emph{Consensus} strategy presented in Subsection~\ref{Math}.
Additionally, Figure~\ref{fig:simdb} presents the decoding performance in terms of FER for different values of \(E_b/N_0\) at iteration \(i = 10\).

To demonstrate that the \emph{Return} strategy is necessary to maintain the performance improvement, we included the curves corresponding to the case in which only the \emph{Consensus} strategy is applied, referred to as C-RBP (dashed cyan curve in Figure~\ref{fig:sim1}).
Observe that C\&DR-RBP behaves identically to C-RBP during the initial phase (\(i < i_0\)).
However, the importance of the return mechanism becomes evident at later iterations, as C-RBP exhibits a degraded performance for \(i > i_0\) compared with both C\&R-RBP and C\&DR-RBP\@.
This result shows that, although the \emph{Return} step may not yield the best  performance during the first decoding iterations, as suggested by the performance of C\&R-RBP, it is essential to maintain improved decoding performance at later iterations.

\subsection{Simulation 2: Code Length = 1248, R = 1/2}

We repeat the same simulation framework for a rate-1/2, length-1248 code (base graph BG2, lifting size \(Z = 72\), and set index \(i_{LS} = 4\)), and \(E_b/N_0 = 2.0\)~dB.
The corresponding results are shown in Figure~\ref{fig:sim2}.
For C\&DR-RBP, the \emph{Return} strategy is activated again at iteration \(i_0 = 3\).
As shown in Figure~\ref{fig:sim2}, the results are qualitatively similar to those obtained for the shorter code length.
C\&R-RBP outperforms the competing algorithms for \(i \ge 3\), whereas C\&DR-RBP achieves superior performance at every iteration while requiring fewer iterations to achieve a given error-rate performance.
Figure~\ref{fig:simdb2} additionally presents the decoding performance in terms of FER for different values of \(E_b/N_0\) after 10 iterations for the rate-1/2 length-1248 code.

\subsection{Simulation 3: Code Length = 1248, R = 2/3}

We also evaluated the proposed algorithms using a rate-2/3 length-1248 code (base graph BG1, lifting size \(Z = 40\), and set index \(i_{LS} = 2\)) and \(E_b/N_0 = 2.8\)~dB.
The corresponding results are presented in Figure~\ref{fig:sim3}.
For C\&DR-RBP, the \emph{Return} strategy is activated at iteration \(i_0 = 5\), which yields the best performance.
As shown in Figure~\ref{fig:sim3}, the C\&R-RBP again outperforms the competing algorithms for \(i \ge 3\).
On the other hand, C\&DR-RBP achieves superior performance at nearly every iteration compared with the competing algorithms.
The only exception occurs at iteration \(i = 5\) in terms of FER, where its performance numerically matches that of SVNF\@.
However, C\&DR-RBP achieves superior BER performance at every iteration.
Figure~\ref{fig:simdb3} additionally presents the decoding performance in terms of FER for different values of \(E_b/N_0\) after 10 iterations for the rate-2/3 length-1248 code.

\subsection{Simulation 4: Code Length = 576, R = 5/6}

In this setting, we simulated a rate-5/6 length-576 code (base graph BG1, lifting size \(Z = 22\), and set index \(i_{LS} = 5\)) and \(E_b/N_0 = 5.0\)~dB.
The corresponding results are presented in Figure~\ref{fig:sim4}.
In this high-SNR high-rate scenario, C\&R-RBP does not perform well for the first iterations in terms of FER (\(i \le 3\)), but its error-rate performance improves for \(i \ge 4\) for FER and \(i \ge 3\) for BER.
On the other hand, C\&DR-RBP, with \(i_0 = 4\) chosen heuristically, shows again the most competitive performance both in terms of FER and BER.
Figure~\ref{fig:simdb3} additionally presents the decoding performance in terms of FER for different values of \(E_b/N_0\) after 10 iterations for the rate-5/6 length-576 code.

\subsection{Simulation 5: Code Length = 576, R = 1/2 (Fading Channel)}

We also consider a more challenging scenario involving an uncorrelated Rayleigh fading channel model, in which the received signal is given by \(y_j = \gamma_j \, u_j + w_j\), where
\(\gamma_j\) is a realization of the random variable \(\Gamma_j = \sqrt{X_1^2 + X_2^2}\),
with \(X_1 \sim \cN(0,1)\) and \(X_2 \sim \cN(0,1)\).
We use the same rate-1/2 length-576 code considered in Simulation~1, with \(E_b/N_0 = 6.0\)~dB.
The corresponding results are presented in Figure~\ref{fig:sim5}.
For C\&DR-RBP, the \emph{Return} strategy is activated at iteration \(i_0 = 3\).
As shown in Figure~\ref{fig:sim5}, C\&R-RBP once again outperforms the competing algorithms for \(i \ge 3\).
In contrast, C\&DR-RBP achieves superior performance at every iteration compared with the competing algorithms, both in terms of FER and BER\@.
Figure~\ref{fig:simdb5} additionally presents the decoding performance in terms of FER for different values of \(E_b/N_0\) after 10 iterations for the rate-1/2 length-576 code under the Rayleigh fading channel model.

\subsection{Comparison Tables}

To help compare the performance results, we included in Tables~\ref{tab:FER} and \ref{tab:BER} the number of iterations necessary for each RBP-based algorithm to achieve a given error-rate performance for each of the simulation settings considered (denoted by S1 up to S5).
For example, column S1 of Table~\ref{tab:FER} shows the number of iterations needed for each algorithm to achieve a FER \(< 10^{-4}\) in Simulation~1.
The lowest values are highlighted.

\begin{table}[ht]
    \centering\footnotesize
    \caption{Number of iterations to achieve a given FER.}
    \begin{tabular}{cccccc}
        \toprule
         & S1 & S2 & S3 & S4 & S5 \\
         \midrule
         FER & \(10^{-4}\) & \(10^{-3}\) & \(10^{-3}\) & \(10^{-4}\) & \(10^{-4}\) \\
         \midrule
         RBP        & \(>\!\!20\) & 11 & \(>\!\!20\) & \(>\!\!20\) & \(>\!\!20\)\\
         RD-RBP     & 10 & 9 & 9 & 8 & 9\\
         NW-RBP     & 13 & 13 & 17 & 15 & 13\\
         SVNF       & 6  & 5 & \textbf{4} & 9 & 5 \\
         List-RBP   & 7  & 5 & 6 & \textbf{6} & 5\\
         UBP-RBP    & \(>\!\!20\) & 5 & 5 & \(>\!\!20\) & \(>\!\!20\) \\
         C\&R-RBP   & \textbf{5}  & \textbf{4} & \textbf{4} & \textbf{6} & \textbf{4}\\
         C\&DR-RBP  & \textbf{5} & \textbf{4} & \textbf{4} & \textbf{6} & 5\\
         \bottomrule
    \end{tabular}
    \label{tab:FER}
\end{table}

\begin{table}[!ht]
    \centering\footnotesize
    \caption{Number of iterations to achieve a given BER.}
    \begin{tabular}{cccccc}
        \toprule
         & S1 & S2 & S3 & S4 & S5 \\
         \midrule
         BER & \(10^{-5}\) & \(10^{-4}\) & \(10^{-4}\) & \(10^{-5}\) & \(10^{-5}\) \\
         \midrule
         RBP        & \(>\!\!20\) & 8 & 6 & \(>\!\!20\) & \(>\!\!20\)\\
         RD-RBP     & 10 & 7 & 5 & 4 & 10\\
         NW-RBP     & 11 & 10 & 10 & 9 & 13\\
         SVNF       & 6  & 4 & 3 & 5 & 6 \\
         List-RBP   & 6  & 4 & 3 & \textbf{4} & 6\\
         UBP-RBP    & \(>\!\!20\) & \textbf{3} & 3 & \(>\!\!20\) & \(>\!\!20\) \\
         C\&R-RBP   & \textbf{4}  & \textbf{3} & \textbf{2} & \textbf{4} & \textbf{4}\\
         C\&DR-RBP  & \textbf{4}  & \textbf{3} & \textbf{2} & \textbf{4} & 5\\
         \bottomrule
    \end{tabular}
    \label{tab:BER}
\end{table}  

    \section{Decoding Latency and Memory Analysis}\label{Latency}

        Following the model in \cite{Celebi2021}, we consider the aggregate latency to consist of three main components: \textit{i}) latency at the transmitter, \textit{ii}) transmission and propagation latencies, and \textit{iii}) latency at the receiver.
Transmission latency is the time required to send all the symbols of a codeword over the channel.
The latencies at the transmitter and receiver can be further divided into latency due to encoding, buffering, signal processing, interleaving, and decoding.
Except for encoding and decoding, these latencies are independent of the particular coding/decoding scheme; therefore, they contribute a constant term that is neglected in the analysis that follows.
As discussed earlier, this work focuses on reducing the decoding latency of IDS methods by reducing the number of algorithm iterations required to achieve a given error-rate performance.

Given the dynamic nature of the RBP-based scheduling considered in this work, we focus our latency analysis on a GPP-oriented implementation model, which provides a convenient framework for estimating the decoding latency of the proposed algorithms.
Although the decoding latency depends on several factors, such as the particular hardware platform and preprocessing operations, we follow \cite{Celebi2021} and assume that the decoding operations are handled sequentially by a processor.
Under this assumption, the decoding time is linearly related to the processor clock-cycle duration, \(1/f_\text{clk}\), where \(f_{\text{clk}}\) is the processor clock frequency.
Thus, the decoding latency \(L_d\) can be estimated as
\begin{equation}\label{latency1}
    L_d = \frac{I \muspace N_{\text{ops}} }{f_{\text{clk}}} \times \text{PCO},
\end{equation}
where \(I\) is the number of decoding iterations, \(N_{\text{ops}}\) is the total average number of operations per iteration\footnote{Here, an operation refers broadly to either an arithmetic operation or a memory access.}, and PCO is the \emph{average} number of \emph{processor cycles per operation} for a given processor.
In other words, PCO is an aggregate weighted average that approximately accounts for the different execution costs of operations, such as additions, multiplications, and memory accesses.
This avoids explicitly modeling the number of processor cycles required by each individual operation.
Since our goal is to compare the decoding latencies of different algorithms under the same framework, we set PCO \(=1\) throughout the following analysis to simplify the calculations.
Therefore, the decoding latencies for any other value of PCO can be obtained by multiplying the reported results by that value.

\subsection{Estimating the Number of Operations}\label{Nops}

\subsubsection{C2V and Residual Calculation}

Let us consider the alternative definition of C2V messages given in \eqref{c2valt}.
Thus, the C2V evaluation can be divided into the following steps.
For each edge \((c_i,v_j)\):
\begin{enumerate}
    \item Read \((\bard_c - 1)\) V2C messages from memory.
    \item Compute the absolute values (\(\beta_{v_k}\)) and signs (\(\alpha_{v_k}\)) of these messages.
    \item Calculate \(\phi(\beta_{v_k})\) for each V2C message.
    \item Compute the sum of the \((\bard_c - 1)\) quantities \(\phi(\beta_{v_k})\).
    \item Calculate \(\phi(\cdot)\) over the result from the last step.
    \item Compute the product of the \((\bard_c - 1)\) signs \(\alpha_{v_k}\).
    \item Multiply the results of steps 5 and 6 and store the result in memory.
\end{enumerate}

The corresponding number of operations, considering that \(\phi(\cdot)\) requires \(K\) operations, is summarized in Table~\ref{tab:C2V}.

\begin{table}[ht]
    \centering\footnotesize
    \caption{Estimated number of operations for C2V evaluation.}
    \begin{tabular}{cc}
        \toprule
         Step & Number of operations \\
         \midrule
         1 & \(\bard_c - 1\) \\
         2 & \(2(\bard_c-1)\) \\
         3 & \(K(\bard_c-1)\) \\
         4 & \(\bard_c-2\) \\
         5 & \(K\) \\
         6 & \(\bard_c-2\) \\
         7 & \(2\) \\
        \midrule
         Total & \((K + 5)\bard_c - 5\) \\
         \bottomrule
    \end{tabular}
    \label{tab:C2V}
\end{table}

To calculate residuals, consider the following additional steps: (\textit{i}) read \(\cv{i}{j}\) from memory, (\textit{ii}) calculate \(\residual{i}{j}\) using \eqref{res}, and (\textit{iii}) write \(\residual{i}{j}\) to memory.
These steps add 4 operations to the total number in Table~\ref{tab:C2V}, resulting in
\begin{equation}\label{res_ops}
    (K + 5)\bard_c - 1.
\end{equation}

\subsubsection{Approximation of Function $\phi$}

Evaluating \(\phi(\cdot)\) is computationally expensive and may require significant hardware resources.
To reduce the implementation complexity and maintain the calculation accuracy as much as possible, an approximation of \(\phi(\cdot)\) is imperative.
For example, \cite{Yang2021} presents a non-uniform piecewise approximation scheme using the quadratic function \(y = ax^2 + bx + c\), which achieves decoding performance almost identical to that of the standard SPA with appropriate determination of parameters \(a\), \(b\) and \(c\) for each segment of the domain \((0,\infty)\).

Let \(S\) be the number of segments, the number of operations to evaluate \(\phi(x)\) can be estimated using the following steps:
\begin{enumerate}
    \item At most \(S-1\) comparisons to decide in which segment \(x\) lies and \(S-1\) memory accesses to recover the bounds of the segments.
    \item A total of 5 operations (3 multiplications, 2 sums) to evaluate \(ax^2 + bx + c\,\).
    \item 3 memory accesses for \(a\), \(b\), and \(c\).
    \item 1 memory write to store the result.
\end{enumerate}
Thus, considering the worst-case scenario in finding the appropriate segment of \(x\), we have
\begin{equation}\label{K}
    K=2(S-1)+9 = 2S+7.
\end{equation}

\subsubsection{Real Value Comparisons}\label{realvalue}

In order to find the maximum residual at each algorithm loop in RBP-based algorithms, the theoretical number of real-value comparisons is \(E-1\) (see Table~\ref{tab:complexity}).
Accounting for \(E\) memory accesses, this results in \(2E - 1\) operations.
However, this number can be significantly reduced in a practical implementation if we record the maximum residual for each check node \(c_i\muspace\) (similarly to NW-RBP).
In other words, if we record
\begin{equation}\label{omega}
    \omega_i = \max_{v_j \in \cV(c_i)} \residual{i}{j}.
\end{equation}
Given \(\omega_i\muspace\) for each \(c_i\), we can perform the following steps to find the maximum residual and maintain \(\omega_i\) updated:
\begin{enumerate}
    \item Read all \(\omega_i\) from memory.
    \item Find \(i^* = \argmax_i \omega_i\muspace\).
    \item Read, for all \(v_j \in \cV(c_{i^*})\), \(\residual{i^*}{j}\) from memory.
    \item Find the maximum residual among \(\residual{i^*}{j}\).
    \item Find the new \(\omega_i\) considering  only \(v_j \in \cV(c_{i^*})\setminus\{v_{j^*}\}\), and store it in memory.
    \item To maintain the values of \(\omega_i\) up to date, compare every newly calculated residual with the corresponding \(\omega_i\muspace\).
\end{enumerate}
The corresponding number of operations is given in Table~\ref{tab:comparison_calc}.
\begin{table}[ht]
    \centering\footnotesize
    \caption{Estimated number of operations to find the largest residual.}
    \begin{tabular}{cc}
        \toprule
        Step & Number of operations \\
        \midrule
        1 & \(M\) \\
        2 & \(M-1\) \\
        3 & \(\bard_c\) \\
        4 & \(\bard_c-1\) \\
        5 & \(\bard_c-1\) \\
        6 & \((\bard_v-1)(\bard_c+1)\) \\
        \midrule
        Total & \(2M + (\bard_c + 1)(\bard_v+2) - 6\) \\
        \bottomrule
    \end{tabular}
    \label{tab:comparison_calc}
\end{table}

Step 6 above requires a more detailed description.
If the newly calculated residual is greater than \(\omega_i\muspace\), then \(\omega_i\muspace\) must be updated with this new value.
Thus, for each \(c_i\muspace\) (see Line~11 in Algorithm~\ref{alg1}), \(\bard_c-1\) comparisons with \(\omega_i\muspace\) are required, corresponding to the number of newly calculated residuals shown in Fig.~\ref{Fig:RBP}.
In addition, two memory accesses are required to read \(\omega_i\muspace\) and write its updated value back to memory.
Thus, each \(c_i\) contributes \(\bard_c+1\) operations.
Therefore, for each C2V update, we have \((\bard_v-1)(\bard_c+1)\) additional comparisons.

\subsubsection{Global estimate for RBP}

For the global estimate of the number of operations per algorithm loop, we consider the following:
\begin{enumerate}
    \item The V2C updates, which require, according to \eqref{v2c}, (\textit{i}) \(\bard_v\) memory accesses, (\textit{ii}) \(\bard_v - 1\) additions, and (\textit{iii}) one memory write, resulting in a total of \(2\bard_v\) operations.
    \item The C2V update, given in the last row of Table~\ref{tab:C2V}.
\end{enumerate}
The second column of Table~\ref{tab:global_calc} summarizes the estimated average number of operations for RBP.
Now, considering a \(576\)-code with \(R = 1/2\), according to the 3GPP procedure in \cite{3GPP2025}, this code configuration yields\footnote{Note that \(N \ne 576\) and \(M \ne 288\) due to the resulting \(2\muspace Z = 80\) punctured nodes in the 5G NR parity-check matrix.}
\begin{itemize}
    \item \(M = 352\), \(N = 656\), and \(E = 2088\).
    \item \(\bard_c = 5.9319\) and \(\bard_v = 3.1829\).
\end{itemize}
Also, taking \(S\! =\! 4\) in \eqref{K} as suggested in \cite{Yang2021}, we have \(K \!=\! 15\) operations.
Substituting these values into the second column of Table~\ref{tab:global_calc}, we arrive at the numerical values displayed in the third column.
Note that residual calculations and real-value comparisons are the most significant contributors to the number of operations.
Also note that the total number of real-value comparisons is significantly smaller than the theoretical value \(2E-1 = 4175\), as discussed.

\begin{table}[ht]
    \centering\footnotesize
    \caption{Estimated number of operations for RBP.}
    \begin{tabular}{ccc}
        \toprule
        Operation & Estimated number of operations & Results \\
        \midrule
        C2V update & \((K + 5)\bard_c - 5\) & \(113.64\)\\
        V2C updates & \(2\bard_v (\bard_v-1)\) & \(13.90\)\\
        Residuals & \(((K + 5)\bard_c - 1) (\bard_v-1)(\bard_c-1)\) & \(1266.45\) \\
        Comparisons & \(2M + (\bard_c + 1)(\bard_v+2) - 6\) & \(733.93\)\\
        \midrule
        Total &  & \(2127.91\) \\
        \bottomrule
    \end{tabular}
    \label{tab:global_calc}
\end{table}

\subsubsection{Global estimate for RD-RBP}

In addition to the operations required for the C2V update, RD-RBP must account for the update of the decay coefficient \(\beta_{i,j}\muspace\).
Thus, additional \(3\) operations are needed in the C2V update to (\textit{i}) read \(\beta_{i,\muspace j}\) from memory, (\textit{ii}) multiply it by \(\beta\) and (\textit{iii}) write it back to memory.
Hence, from the total of operations in Table~\ref{tab:C2V}, we have
\begin{equation}\label{res_ops_beta}
    (K + 5)\bard_c - 2.
\end{equation}
Also, for each residual evaluation, \(\beta_{i,\muspace j}\) must be read from memory and multiplied by \(\residual{i}{j}\).
Thus, additional \(2\) operations are added to \eqref{res_ops}, resulting in
\begin{equation}\label{res_ops_decay}
    (K + 5)\bard_c + 1.
\end{equation}        
The estimated average number of RD-RBP operations and the estimated numerical results are shown in Table~\ref{tab:RD_calc}.

\begin{table}[ht]
    \centering\footnotesize
    \caption{Estimated number of operations for RD-RBP.}
    \begin{tabular}{ccc}
        \toprule
        Operation & Estimated number of operations & Results \\
        \midrule
        C2V update & \((K + 5)\bard_c - 2\)         & \(116.64\)\\
        V2C updates & \(2\bard_v (\bard_v-1)\) & \(13.90\)\\
        Residuals & \(((K + 5)\bard_c + 1)(\bard_v-1)(\bard_c-1)\) & \(1287.98\) \\
        Comparisons & \(2M + (\bard_c + 1)(\bard_v+2) - 6\) & \(733.93\)\\
        \midrule
        Total &  & \(2152.44\) \\
        \bottomrule
    \end{tabular}
    \label{tab:RD_calc}
\end{table}  

\subsubsection{List-RBP}

For List-RBP, we need to take into account the operations necessary to maintain the main list updated.
First, at each iteration, the edge corresponding to the largest residual must be removed from the top of the list.
Then, the list is reordered.
We estimate \(4(S_1-1)\) operations, which correspond to the total number of reads from and writes to memory of each edge \((c_i,v_j)\) and each residual \(\residual{i}{j}\) on the list.

Second, each time a new residual \(\residual{i}{j}\) is calculated, we must check whether the corresponding edge \((c_i,v_j)\) is already on the main list.
If so, the edge is removed from the list and this is reordered accordingly.
Let \(q\) be the position of the residual on the list; thus, there must be \(q\) reads from memory and \(q\) comparisons to find the position.
Then, \(S_1 - q\) reads and \(S_1 - q\) writes to memory are necessary to reorder the list, making a total of \(2S_1\) operations.
Note that if the edge is not on the main list, the total number of operations is also \(2S_1\).
Therefore, \(2S_1(\bard_v-1)(\bard_c-1)\) operations are necessary.

Table~\ref{tab:list_calc} shows the estimated number of operations for List-RBP and the corresponding numerical results, using the same list sizes as in our simulations (\(S_1 = 16, S_2 = 2\)).

\begin{table}[ht]
    \centering\footnotesize
    \caption{Estimated number of operations for List-RBP.}
    \begin{tabular}{ccc}
        \toprule
        Operation & Estimated number of operations & Results \\
        \midrule
        C2V update & \((K + 5)\bard_c - 2\) & \(116.64\)\\
        V2C updates & \(2\bard_v(\bard_v-1)\) & \(13.90\)\\
        Residuals & \(((K + 5)\bard_c + 1)(\bard_v-1)(\bard_c-1)\) & \(1287.98\) \\
        Comparisons &  \( 2 \muspace S_2\big((\bard_v-1)(\bard_c-1) + S_1\big)\) & \(107.06\)\\
        Update List & \(4(S_1-1) + 2S_1(\bard_v-1)(\bard_c-1)\) & \(404.51\)\\
        \midrule
        Total & & \(1930.08\) \\
        \bottomrule
    \end{tabular}
    \label{tab:list_calc}
\end{table}      

\subsubsection{C\&R-RBP and C\&DR-RBP estimates}

Since both algorithms store the precomputed C2V messages used to evaluate residuals, the C2V update requires:
\begin{enumerate}
    \item \(d_v\) memory accesses to retrieve the precomputed C2V messages \(\cv{i}{j}^\prime\) for \emph{Consensus}.
    \item \(d_v\) memory writes to store these values in the newly updated C2V messages \(\cv{i}{j}\) .
    \item \(3\bard_v\) operations to update the decay coefficients \(\beta_{i,k}\).
\end{enumerate}
Thus, the total number of operations is \(5\bard_v\) for the C2V update.
Additionally, we have to account for \emph{Return}:
\begin{enumerate}
    \item one additional V2C update (the \emph{returning} message), which accounts for additional \(2\bard_v\) operations.
    \item \(\bard_c-1\) additional residual calculations.
    \item \(\bard_c + 1\) additional comparisons to update \(\omega_{i^*}\) (see the discussion of step~6 in Subsection~\ref{realvalue}).
\end{enumerate}
Finally, we must account for the storage of the precomputed C2V messages \(\cv{i}{j}^\prime\); thus, the total number of operations for each residual calculation is, from \eqref{res_ops_decay},
\[
    \bard_c(K + 5) + 2.
\]

Table~\ref{tab:CR_calc} shows the estimated number of operations and the corresponding results for C\&R-RBP\@.
As expected, C\&R-RBP requires a larger number of operations than similar RBP-based IDS methods.
However, as discussed in Subsection~\ref{Complexity}, the number of algorithm loops is reduced from \(E\) to \(E-N\) in order to reduce the average number of residual calculations.
When we consider the algorithms within a complete iteration, we arrive at the values shown in Table~\ref{tab:placeholder}.
As we can see, C\&R-RBP reduces, on average, the total number of operations per iteration.

\begin{table}[ht]
    \centering\footnotesize
    \caption{Estimated number of operations for C\&R-RBP.}
    \begin{tabular}{ccc}
        \toprule
        Operation & Estimated number of operations & Results \\
        \midrule
        C2V update & \(5\bard_v\) & \(15.91\)\\
        V2C updates & \(2\bard_v\muspace\bard_v\) & \(20.26\)\\
        Residuals & \((\bard_c(K + 5)+2) \bard_v(\bard_c-1)\) & \(1893.70\) \\
        Comparisons & \(2M + (\bard_c + 1)(\bard_v+3) - 6\) & \(739.86\)\\
        \midrule
        Total &  & \(2669.74\) \\
        \bottomrule
    \end{tabular}
    
    \label{tab:CR_calc}
\end{table}

\begin{table}[ht]
    \centering\footnotesize
    \caption{Number of operations per algorithm iteration.}
    \begin{tabular}{cc}
        \toprule
        Algorithm & Operations per iteration \\
        \midrule
        RBP & \(E\times 2127.91 = 4\,443\,074\)\\
        RD-RBP & \(E\times 2152.44 = 4\,494\,296\)\\
        List-RBP & \(E\times 1930.08 = 4\,030\,012\)\\
        C\&R-RBP & \((E-N)\times 2669.74  = 3\,823\,066\) \\
        C\&DR-RBP (\(i \le i_0\)) & \(E \times 2062.48  = 4\,306\,468\) \\
        \bottomrule
    \end{tabular}
    \label{tab:placeholder}
\end{table}

Finally, Table~\ref{tab:global_calc_CDR} shows the estimates for C\&DR-RBP for \(i < i_0\).
Accordingly, the additional operations due to \emph{Return} are not considered.
\begin{table}[ht]
    \centering\footnotesize
    \caption{Estimated number of operations for C\&DR-RBP (\(i < i_0)\).}
    \begin{tabular}{ccc}
        \toprule
        Operation & Estimated number of operations & Results \\
        \midrule
        C2V update & \(5\bard_v\) & \(15.91\)\\
        V2C updates & \(2\bard_v(\bard_v-1)\) & \(13.90\)\\
        Residuals & \((\bard_c(K + 5) + 2)(\bard_v-1)(\bard_c-1)\) & \(1298.75\) \\
        Comparisons & \(2M + (\bard_c + 1)(\bard_v+2) - 6\) & \(733.93\)\\
        \midrule
        Total & \(\phantom{\big(\big)}\) & \(2062.48\) \\
        \bottomrule
    \end{tabular}
    
    \label{tab:global_calc_CDR}
\end{table}

\subsubsection{Other RBP-based algorithms}

Since NW-RBP, SVNF, and UBP-RBP do not differ significantly from RBP in terms of number of operations, we consider, for simplicity, that they have the same \(N_{\text{ops}}\) of RBP.

\subsection{Estimating the Latency}  

Following \cite{Touati2025}, we consider an ARM Cortex-A53 processor as a representative low-power processor architecture with a clock frequency \(f_{\text{clk}} = 1.2\) GHz.
To estimate the decoding latency, we consider the number of algorithm iterations required to achieve a BER \(\le 10^{-5}\) at \(E_b/N_0 = 3.0\)~dB.
From Figure~\ref{fig:sim1}, we can determine the number of iterations required by each algorithm to achieve this benchmark.
In this analysis, we consider only RBP-based IDS algorithms.
The required number of iterations is reported in the second column of Table~\ref{tab:latency}.
Using these values, Eq.~\eqref{latency1} and the corresponding \(N_{\text{ops}}\) for each algorithm\footnote{For C\&DR-RBP, the corresponding \(N_{\text{ops}}\) was calculated as a weighted mean for \(i < i_0\) and \(i \ge i_0\): \((2\times  3 406 468 + 2\times 823 066)/4 = 4 064 767\).}, we can estimate the total decoding latency, shown in the third column of Table~\ref{tab:latency}\footnote{The precomputation of the residuals (lines~3 and~4 in Algorithm~\ref{alg1}), performed by all the RBP-based algorithms, is also considered and contributes \(E((K + 5)\bard_c - 3)\) operations (the total C2V update operations in Table~\ref{tab:C2V} plus (\textit{i}) taking the absolute value and (\textit{ii}) writing the result to memory).
The total estimated contribution to decoding latency is \(0.2\)~ms.}.
Note that both C\&R-RBP and C\&DR-RBP show a significant reduction in latency compared with the most competitive algorithm in this scenario, namely List-RBP: a reduction of 36\% and 33\%, respectively.

\begin{table}[ht]
    \centering\footnotesize
    \caption{Number of iterations to achieve a BER \(\le 10^{-5}\) for each RBP-based algorithm, considering a 5G NR \(576\)-code with \(R = 1/2\) at \(E_b/N_0 = 3.0\)~dB (see Figure~\ref{fig:sim1}).}
    \begin{tabular}{cccc}
        \toprule
        Algorithm & \(N_{\text{ops}}\) & Iterations & Latency (ms) \\
        \midrule
        RBP          & 4 443 074 & \(>\)20  & \(>\)74.25  \\
        RD-RBP       & 4 494 296 & 9  & 33.91\\
        NW-RBP       & 4 443 074 & 11 & 40.93\\
        SVNF         & 4 443 074 & 7  & 26.12\\
        List-RBP     & 4 030 012 & 6  & 20.35\\
        UBP-RBP      & 4 443 074 & \(>\)20  & \(>\)74.25  \\
        C\&R-RBP     & 3 823 066 & 4  & 12.94\\
        C\&DR-RBP    & 4 064 767 & 4  & 13.55\\
        \bottomrule
    \end{tabular}
    \label{tab:latency}
\end{table}

We also considered the same decoding setting, but for the number of iterations required to achieve a BER \(\le 10^{-4}\).
Again, from Figure~\ref{fig:sim1}, we can determine the required number of iterations and calculate the corresponding latencies.
The results are given in Table~\ref{tab:latency2}.
Now we can see that C\&DR-RBP shows a more significant reduction in latency (42\%) compared to the most competitive algorithm in this scenario, namely C\&R-RBP.

\begin{table}[ht]
    \centering\footnotesize
    \caption{Number of iterations to achieve a BER \(\le 10^{-4}\) for each RBP-based algorithm, considering  a 5G NR \(576\)-code with \(R = 1/2\) at \(E_b/N_0 = 3.0\)~dB (see Figure~\ref{fig:sim1}).}
    \begin{tabular}{cccc}
        \toprule
        Algorithm & \(N_{\text{ops}}\) & Iterations & Latency (ms) \\
        \midrule
        RBP       & 4 443 074 & 3 & 11.31  \\
        RD-RBP    & 4 494 296 & 3 & 11.44  \\
        NW-RBP    & 4 443 074 & 5 & 18.71 \\
        SVNF      & 4 443 074 & 2 & 7.60  \\
        List-RBP  & 4 030 012 & 2 & 6.92  \\
        UBP-RBP   & 4 443 074 & 2 & 7.60  \\
        C\&R-RBP  & 3 823 066 & 2 & 6.57  \\
        C\&DR-RBP & 4 306 468 & 1 & 3.79  \\
        \bottomrule
    \end{tabular}
    \label{tab:latency2}
\end{table}

\subsection{Memory Requirements}\label{memory}

In general, RBP-based algorithms require memory storage for
\begin{enumerate}
    \item V2C messages (\(E\) real values).
    \item C2V messages (\(E\) real values).
    \item C2V residuals (\(E\) real values).
    \item Current LLRs (\(N\) real values).
    \item Previous LLRs (\(N\) real values).
    \item Syndrome values: \(M\) binary values.
    \item Parity-check matrix \(\bmH\): \(MN\) binary values.
    \item Message estimate: \(N\) binary values.
    \item Extra \(M\) real values to store \(\omega_i\) in \eqref{omega}.
\end{enumerate}

Assuming a 64-bit floating-point representation, the total memory requirement is approximately
\begin{equation*}
    M_{\text{RBP}} = 8(3E + 2N + M) + M N + M + N.
\end{equation*}
For List-RBP, (3) is replaced by \(S_1 + S_2\) stored values for the two lists, and (9) is not necessary.
Therefore, for List-RBP,
\begin{equation*}
    M_{\text{List-RBP}} = 8(2E + 2N + S_1 + S_2) + M N + M + N.
\end{equation*}
Considering a 576-code with rate \(1/2\), according to the 3GPP procedure described in \cite{3GPP2025}, we have \(M = 352\), \(N = 656\), and \(E = 2088\).
Therefore, \(M_{\text{RBP}} \approx 288\) KB and \(M_{\text{List-RBP}} \approx 269.5\) KB.
On the other hand, for C\&R-RBP and C\&DR-RBP, the precomputed C2V messages \(\cv{i}{j}^\prime\) are stored for subsequent updates.
Thus, additional \(E\) real values are necessary, resulting in approximately 305 KB.
These results are shown in Table~\ref{tab:memory}.
As we can see, our methods require only 5.9\% more memory than the other RBP-based algorithms and 13.2\% more memory than List-RBP.

\begin{table}[ht]
    \centering\footnotesize
    \caption{Memory requirements for a 5G NR 576-code with rate 1/2.}
    \begin{tabular}{cc}
        \toprule
        Algorithm & Memory requirement (KB) \\
        \midrule
        C\&R-RBP, C\&DR-RBP & 305 \\
        List-RBP & 269.5 \\
        other RBP-based & 288\\
        \bottomrule
    \end{tabular}
    \label{tab:memory}
\end{table}
	
	\section{Conclusions}\label{Conclusions}
	
        In this paper, we introduced the algorithms C\&R-RBP and C\&DR-RBP, which further exploit the availability of updated information to mitigate the performance degradation caused by the greedy RBP behavior and the propagation of errors.
The simulation results demonstrated that both algorithms outperform the IDS RBP-based decoding methods with the same overall computational complexity.
In particular, C\&DR-RBP proved especially suitable for low-latency applications due to its lower number of iteration required for achieving a given error-rate performance and the absence of additional computational cost.
The simulation results also confirm the conclusions of the formal analysis, which showed that the \emph{Consensus} strategy increases the probability of correct bit decisions under the initial decoding conditions and in the presence of severe bit transmission errors.
An important direction for future research is the extension of the proposed methods to more general channel models, including more complex fading scenarios such as Rician fading channels, impulsive noise, correlated channels and other non-Gaussian environments.
	
	\appendices
	
	\section{-- Proof of Theorem 1}\label{appA}
	
        Define the following recursive function:
\begin{equation}\label{Ji}
    G_j(\ell_{j+1},... \,,\ell_n) = \intzeroinf G_{j-1}(\ell_j,... \,,\ell_n)h_j(\ell_j)d\ell_j\muspace,				
\end{equation}
for \(j \in \{1,... \,,n-1\}\), and 
\begin{equation}\label{G0}
    G_0(\ell_1,...\,,\ell_n) = g(\ell_1,...\,,\ell_n).
\end{equation}
Using a finite mathematical induction routine, we can show that \(G_j\) has the following property:
\begin{equation}\label{indution0}
    G_j(... \,,-\ell_i\muspace,...) = -G_j(... \,, \ell_i\muspace ...), \  i \in \{j+1,... \,,n\};
\end{equation}
i.e., \(G_j\) is an odd function with respect to each of its arguments.
The base case \(j = 0\) is trivial, as it follows immediately from \eqref{G0}.
Next, suppose that \eqref{indution0} is true for \(j = j_0\), \(j_0 > 0\).
Thus, for \(j = j_0+1\), from \eqref{Ji} one can easily prove that 
\begin{equation*}
    G_{j_0+1}(\ell_{j_0+2},... \,,-\ell_i\muspace,... \,,\ell_n) = -G_{j_0+1}(\ell_{j_0+2},... \,,\ell_i\muspace,... \,,\ell_n),
\end{equation*}
for \(i \in \{\,j_0+2,... \,,n\}\).
Therefore, \eqref{indution0} is proved.

Next, we show using again a finite mathematical induction routine that
\begin{equation}\label{induction1}
    \E{M} = \intinf f_n(\ell_n)\muspace...\!\!\intinf f_j(\ell_j)G_{j-1}(\ell_j\muspace, ... \,, \ell_n)d\ell_j...\muspace d\ell_n\muspace,
\end{equation}
for \( j \in \{1,... \,,n\}\) and supposing that \(\E{M}\) exists.
First, we show that \eqref{induction1} is true for \(j=1\) (the base case).
Since \(\{L_j\}_{j=1}^n\) are independent random variables, the corresponding multivariate p.d.f.~is given as
\begin{equation*}
    f_{L_1, ... \,, L_n}(\ell_1, ... \,, \ell_n) = \prod_{j=1}^{n}\,f_j(\ell_j).
\end{equation*}
Thus, by the law of the unconscious statistician, the expected value of \(M\) is given as
\begin{IEEEeqnarray*}{rCl}
    \E{M} & = & \intinf ...\intinf g(\ell_1, ... \,, \ell_n)\prod_{j=1}^{n}\,f_j(\ell_j)d\ell_n... \muspace d\ell_1	\\
    & = & \intinf f_n(\ell_n)\muspace...\!\!\intinf f_1(\ell_1)g(\ell_1, ... \,, \ell_n)d\ell_1...\muspace d\ell_n\muspace.
\end{IEEEeqnarray*}
Therefore, since \(g(\ell_1, ... \,, \ell_n) = G_0(\ell_1, ... \,, \ell_n)\), we have shown that \eqref{induction1} is true for the base case (\(j=1\)).

Next, suppose that \eqref{induction1} is true for \(j = j_0\), i.e.,
\begin{equation}\label{k0}
    \E{M} = \!\intinf f_n(\ell_n)\,...\!\!\intinf f_{j_0}(\ell_{j_0})G_{j_0-1}(\ell_{j_0}, ... \,, \ell_n)d\ell_{j_0}...\muspace d\ell_n.
\end{equation}
Call \(I_{j_0}\) the innermost integral in \eqref{k0}.
We can split the integration interval of \(I_{j_0}\) into \(\sR^-\) and \(\sR^+\), resulting in
\begin{IEEEeqnarray*}{rCl}
    I_{j_0} & = & \int_{-\infty}^{0}\!\!f_{j_0}(\ell_{j_0})G_{j_0-1}(\ell_{j_0}, ... \,, \ell_n)d\ell_{j_0} \\
    && \negmedspace{} + \intzeroinf f_{j_0}(\ell_{j_0})G_{j_0-1}(\ell_{j_0}, ... \,, \ell_n)d\ell_{j_0}. \IEEEyesnumber \label{twointegrals}
\end{IEEEeqnarray*}
Call \(I_{j_0}^-\) the first integral above; by making the substitution \(u = -\ell_{j_0}\), we have
\begin{equation*}
    I_{j_0}^- = -\int_{\infty}^{0}f_{j_0}(-u)G_{j_0-1}(-u, ... \,, \ell_n)du\muspace.
\end{equation*}
According to the property of \(G_j\) given in \eqref{indution0}, we have
\begin{IEEEeqnarray*}{rCl}
    I_{j_0}^- & = & \int_{\infty}^{0}f_{j_0}(-u)G_{j_0-1}(u, ... \,, \ell_n)du \\
    & = & -\intzeroinf f_{j_0}(-u)G_{j_0-1}(u, ... \,, \ell_n)du\muspace,
\end{IEEEeqnarray*}
where the last step is due to the swapping of the lower and upper bounds of the integral.
Thus, we can rewrite \eqref{twointegrals}, substituting \(I_{j_0}^-\) and rearranging, as
\begin{equation}\label{Jk0}
    I_{j_0} = \intzeroinf G_{j_0-1}(\ell_{j_0}, ... \,, \ell_n)(\,f_{j_0}(\ell_{j_0})-f_j(-\ell_{j_0}))d\ell_{j_0}.
\end{equation}
Recognizing above \(h_{j_0}(\ell_{j_0}) = f_{j_0}(\ell_{j_0})-f_{j_0}(-\ell_{j_0})\), we see that the right-hand side of \eqref{Jk0} is just \(G_{j_0}(\ell_{j_0+1},... \,,\ell_n)\), according to the definition of \(G_j\) in \eqref{Ji}.
Therefore, the expression in \eqref{k0}, which is a true statement by assumption, can be rewritten, substituting the innermost integral \(I_{j_0}\) by \(G_{j_0}(\ell_{j_0+1},... \,,\ell_n)\), resulting in \(\E{M}\) being equal to
\begin{equation}\label{k0+1}
    \intinf f_n(\ell_n)\,...\!\!\intinf f_{j_0+1}(\ell_{j_0+1})G_{j_0}(\ell_{j_0+1},... \,,\ell_n)d\ell_{j_0+1}...\muspace d\ell_n\muspace.
\end{equation}
Therefore, since the expression in \eqref{k0+1} is equal to the left-hand side of \eqref{induction1} for \(j = j_0+1\), we have shown that if \eqref{induction1} is true for \(j = j_0\), then it must also be true for \(j=j_0+1\), which concludes the proof of \eqref{induction1} by mathematical induction.

Now, since we established the soundness of the expression in \eqref{induction1}, we can express the \(\E{M}\), choosing \(j=n\), as
\begin{equation}\label{EY}
    \E{M} = \intinf f_n(\ell_n)G_{n-1}(\ell_n)d\ell_n.
\end{equation}
Splitting the integration interval into \(\sR^-\) and \(\sR^+\) results in
\begin{equation*}
    \E{M} = \int_{-\infty}^{0}\!\!f_n(\ell_n)G_{n-1}(\ell_n)d\ell_n + \intzeroinf f_n(\ell_n)G_{n-1}(\ell_n)d\ell_n\muspace.
\end{equation*}
Next, make the change of variable \(u = -\ell_n\) in the first integral above, and use the property of \(G_j\) given in \eqref{indution0}.
Finally, rearranging the resulting terms we can rewrite \eqref{EY} as
\begin{equation}\label{EY2}
    \E{M} = \intzeroinf G_{n-1}(\ell_n)h_n(\ell_n)d\ell_n\muspace.
\end{equation}
Using the definition of \(G_j\) in \eqref{Ji}, we can substitute it iteratively in \eqref{EY2} as follows:
\begin{IEEEeqnarray*}{rCL}
    \E{M} & = & \intzeroinf \intzeroinf G_{n-2}(\ell_{n-1}, \ell_n)h_{n-1}(\ell_{n-1})h_n(\ell_n)d\ell_{n-1}d\ell_n \\
    & \vdots & \\
    & = & \intzeroinf ...\intzeroinf g(\ell_1,...\,,\ell_n)\prod_{i}^{n}h_j(\ell_j)d\ell_1...\muspace d\ell_n\muspace.
\end{IEEEeqnarray*}
which completes the proof.
	
	\section{-- Proof of Corollary~1}\label{appB}
	
        From the definition of \(h_j(\ell_j)\) as \(f_j(\ell_j) - f_j(-\ell_j)\), and since \(L_j\) is Gaussian distributed, we have
\begin{IEEEeqnarray*}{rCl}
    h_j(\ell_j) & = & \frac{\exp\left(-\frac{(\ell_j-\mu_j)^2}{2\sigma_j^2}\right)}{\sigma_j\sqrt{2\pi}} - \frac{\exp\left(-\frac{(-\ell_j-\mu_j)^2}{2\sigma_j^2}\right)}{\sigma_j\sqrt{2\pi}} \\
    & = & \frac{\exp\left(\frac{-\ell_j^2-\mu_j^2}{2\sigma_j^2}\right)}{\sigma_j\sqrt{2\pi}}\bigg(\exp\bigg(\tfrac{\ell_j\mu_j}{\sigma_j^2}\bigg) - \exp\bigg(\tfrac{-\ell_j\mu_j}{\sigma_j^2}\bigg)\bigg) \\
    & = & \tfrac{2}{\sigma_j\sqrt{2\pi}}\,\exp\bigg(\tfrac{-\ell_j^2-\mu_j^2}{2\sigma_j^2}\bigg)\sinh\bigg(\tfrac{\ell_j\mu_j}{\sigma_j^2}\bigg).
\end{IEEEeqnarray*}
Call 
\[
w_j(\ell_j) = \tfrac{2}{\sigma_j\sqrt{2\pi}}\,\exp\bigg(\tfrac{-\ell_j^2-\mu_j^2}{2\sigma_j^2}\bigg);
\]
thus, from Theorem~\ref{Theorem}, \(\E{M}\) is equal to
\begin{equation}\label{signs}
    \intzeroinf ...\intzeroinf g(\ell_1,...\,,\ell_n)\prod_{j=1}^{n}w_j(\ell_j)\prod_{j=1}^{n}\sinh\left(\tfrac{\ell_j\mu_j}{\sigma_j^2}\right)d\ell_1...\muspace d\ell_n\muspace.
\end{equation}
Note that, since each interval of integration above is positive (\(\ell_j \ge 0\)), we have, by the property of function \(g\) in Theorem~\ref{Theorem}, that \(g(\ell_1, ... \,, \ell_n) \ge 0\).
Furthermore, \(w_j(\ell_j) \ge 0\) for all \(j\muspace\).
Therefore, the sign of the integrand in \eqref{signs} is completely determined by the product of hyperbolic sines.
More precisely, from the fact that \(s(\sinh(z)) = s(z)\), we have that
\begin{equation*}
    s\!\left(\sinh\!\left(\tfrac{\ell_j\mu_j}{\sigma_j^2}\right)\right) = s\!\left(\tfrac{\ell_j\mu_j}{\sigma_j^2}\right).
\end{equation*}
Since the integral in \eqref{signs} is over \(\ell_j \ge 0\) for all \(j\muspace\), we have
\begin{equation*}
    s\!\left(\tfrac{\ell_j\mu_j}{\sigma_j^2}\right) = s\left(\mu_j\right).
\end{equation*}
Thus, it follows that
\begin{equation*}
    s\bigg(\prod_{j=1}^{n}\sinh\!\left(\tfrac{\ell_j\mu_j}{\sigma_j^2}\right)\bigg) = \prod_{j=1}^{n}s\!\left(\sinh\!\left(\tfrac{\ell_j\mu_j}{\sigma_j^2}\right)\right) =  \prod_{j=1}^{n} s\left(\mu_j\right).
\end{equation*}
As a result, if \(\prod_{j=1}^{n} s\left(\mu_j\right)\) is positive, the integrand in \eqref{signs} is positive and so is the integral.
Similarly, if \(\prod_{j=1}^{n} s\left(\mu_j\right)\) is negative, the integrand is negative and so is the integral, which completes the proof.
	
	\section{-- Proof of Theorem~2}\label{appC}
	
        From property (\(i\)) of function \(\phi\), we have that
\begin{equation*}
    \max_{\ v_k \in \, \cV(c_i)\setminus\{v_j\}\!\!} \phi\left(\abs{\vc{k}{i}}\right) \le \sum_{v_k \in \, \cV(c_i)\setminus\{v_j\}} \phi\left(\abs{\vc{k}{i}}\right)
\end{equation*}
Therefore, from property (\(ii\)), we can write
\begin{equation}\label{phi2}
    \phi\bigg( \sum_{v_k \in \,\! \cV(c_i)\setminus\{v_j\}} \!\!\!\!\!\!\!\phi\left(\abs{\vc{k}{i}}\right)\bigg) \le \phi\!\left(\max_{v_k \in \, \cV(c_i)\setminus\{v_j\}}\!\!\phi\left(\abs{\vc{k}{i}}\right)\!\right)\!.
\end{equation}
Also, from property (\(ii\)), we have the following:
\begin{equation}
    \max_{v_k \in \, \cV(c_i)\setminus\{v_j\}}\phi\left(\abs{\vc{k}{i}}\right) = \phi\!\left(\min_{v_k \in \, \cV(c_i)\setminus\{v_j\}} \abs{\vc{k}{i}}\right)\!.
\end{equation}
Therefore, from property (\(iii\)), we have
\begin{equation}
    \phi\!\left(\max_{v_k \in \, \cV(c_i)\setminus\{v_j\}}\!\!\phi\left(\abs{\vc{k}{i}}\right)\!\right) = \min_{v_k \in \, \cV(c_i)\setminus\{v_j\}} \abs{\vc{k}{i}}. \label{maxmin}
\end{equation}
Thus, from \eqref{c2valt}, \eqref{phi2} and \eqref{maxmin} we have
\begin{IEEEeqnarray*}{rCl}
    \abs{\cv{i}{j}} & = & \phi\bigg( \sum_{v_k \in \cV(c_i)\setminus\{v_j\}} \! \phi(\beta_{v_k})\bigg)\\
    & \le & \phi\!\left(\max_{v_k \in \, \cV(c_i)\setminus\{v_j\}}\!\!\phi\left(\abs{\vc{k}{i}}\right)\!\right) \\
    & = & \min_{v_k \in \, \cV(c_i)\setminus\{v_j\}} \abs{\vc{k}{i}},
\end{IEEEeqnarray*}
which completes the proof.

    \newpage
    \section{-- Real Value Comparisons of List-RBP}\label{appD}

        In their original paper, the authors of \cite{Yan2024} provided the following estimate of real-value comparisons:
\[
    (E-1)(\bard_v-1)(\bard_c-1)(S_1-1),
\]
where \(S_1\) is the size of the main list (denoted \(L_R\) in the original paper).
They argued that this number results from comparing each corresponding edge in \(L_N\) (the secondary or local list) with those in \(L_R\).
However, we found this estimate to be somewhat inaccurate.
First, each newly calculated residual must be compared with the edges currently in the local list \(L_N\) to determine whether it should be included in this list.
This requires \(S_2\) real-value comparisons, where \(S_2\) is the size of the local list.
Hence, for each algorithm iteration, the total number of comparisons required for this operation is
\[
    E(\bard_v-1)(\bard_c-1) S_2.
\]

Second, for each algorithm iteration, the local list \(L_N\) must be used to partially update the main list \(L_R\), which requires comparing each edge in \(L_N\) with the edges in \(L_R\).
This requires \(S_1 \muspace S_2\) comparisons.
Therefore, we estimate the resulting total number of real-value comparisons per algorithm iteration to be
\[
   E (\bard_v-1)(\bard_c-1)S_2 + S_1 \muspace S_2.
\]
    
    \section{-- Sensitivity to \(i_0\)}\label{appE}

        As stated in Subsection~\ref{Simul1}, we chose \(i_0 = 3\) heuristically to obtain a satisfactory performance for C\&DR-RBP.
We used the performance results from Simulation~1 to guide this choice, particularly the FER curves.
Figure~\ref{fig:I0516} presents the FER curves of C\&DR-RBP for \(2 \le i_0 \le 5\).
The curve of C\&R-RBP is also shown, which is equivalent to C\&DR-RBP with \(i_0 = 1\).
As we can see, moving from \(i_0 = 2\) to \(i_0 = 3\) improves the performance at iteration \(i=2\), but results in slightly worse performance at iteration \(i = 3\).
This is compensated by the best performance for iterations \(i \ge 5\).
This trade-off was the criterion used to choose \(i_0 = 3\).
Although there is a performance difference for \(i \ge 5\), this is relatively small, indicating that C\&DR-RBP is more sensitive to the choice of \(i_0\) during the first iterations (\(i_0 < 3\) than during later iterations.

\begin{figure}[ht]
    \centering
    \includegraphics[trim={18 0 8 0}, clip, width=0.9\linewidth]{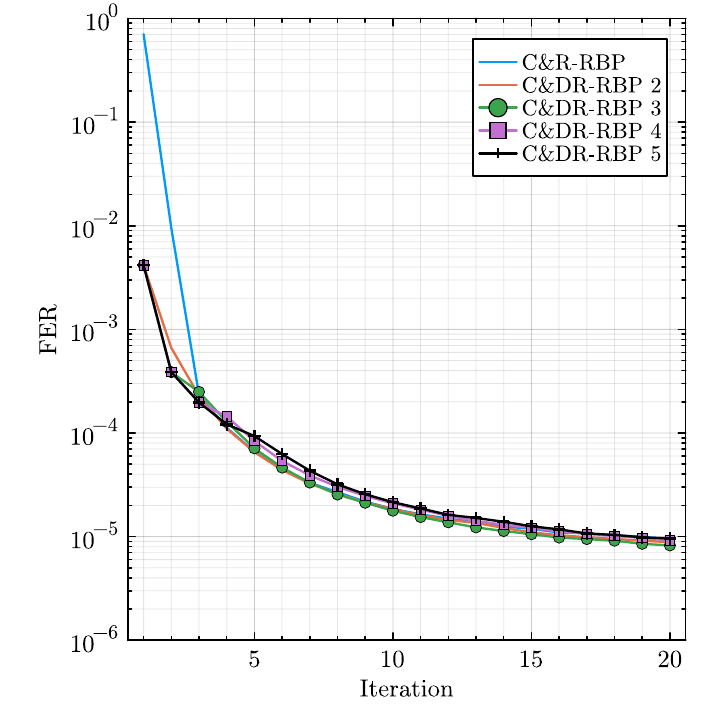}
    \caption{Sensitivity to \(i_0\): FER vs number of iterations using the setup of Simulation~1 for \(i_0 \in \{2,3,4,5\}\).}
    \label{fig:I0516}
\end{figure}

Let us also consider the results obtained using the setup of Simulation~2, shown in Figure~\ref{fig:I01248}.
As we can see, the conclusion reached above also holds here: the sensitivity to \(i_0\) is more pronounced during the first iterations.
It is also worth noting the trade-off associated with  different choices of \(i_0\): while \(i_0 > 3\) yields better performance for \(3 \le i \le 4\), choosing \(i_0 = 3\) yields better performance for \(5 \le i \le 8\).
For \(i > 8\), this choice results in slightly worse performance, but the difference is negligible in practical terms.

\begin{figure}[ht]
    \centering
    \includegraphics[trim={18 0 8 0}, clip, width=0.9\linewidth]{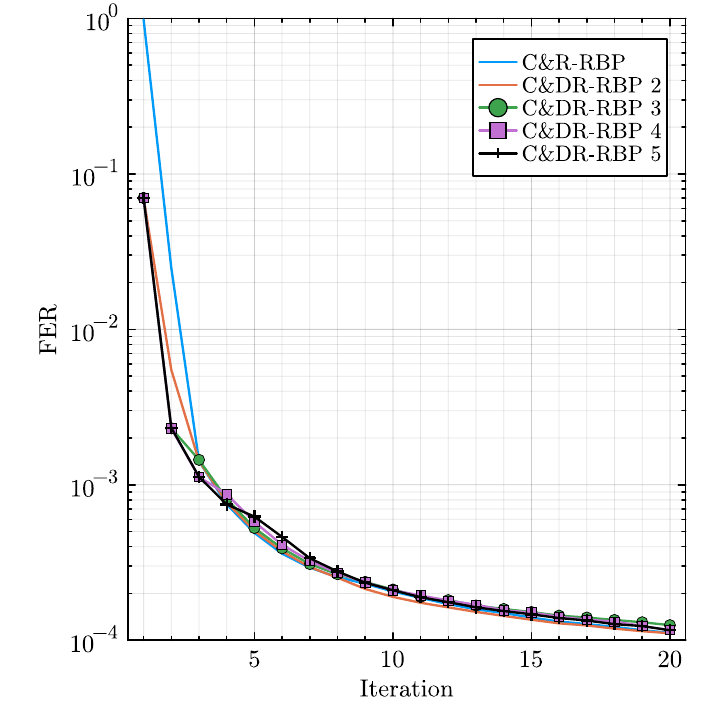}
    \caption{Sensitivity to \(i_0\): FER vs number of iterations using the setup of Simulation~2 for \(i_0 \in \{2,3,4,5\}\).}
    \label{fig:I01248}
\end{figure}
	
	\printbibliography[title=REFERENCES]

	
\end{document}